%% file: main.tex
\documentclass[10pt]{article}

\usepackage{amsmath,amssymb,amsthm}
\usepackage{mathtools}

\usepackage{graphicx}
\usepackage{booktabs}
\usepackage{array}
\usepackage{placeins}

\usepackage[dvipsnames]{xcolor}
\usepackage{colortbl}
\definecolor{mint}{HTML}{CFF2E0}

\usepackage{caption}

\usepackage{hyperref}
\usepackage[capitalize,nameinlink]{cleveref}

\usepackage[backend=biber,style=numeric-comp,sorting=none,giveninits=true,sortcites=true,isbn=false]{biblatex}
\DeclareSourcemap{
  \maps[datatype=bibtex]{
    \map{
      \step[fieldsource=doi, match=\regexp{\A10\.48550/}, final]
      \step[fieldset=doi, null]
    }
    \map{
      \step[fieldsource=archiveprefix, match=\regexp{\A(Biophysics|Bioinformatics)\Z}, final]
      \step[fieldset=archiveprefix, null]
    }
    \map{
      \step[fieldsource=publisher, match=\regexp{\A(Biophysics|Bioinformatics)\Z}, final]
      \step[fieldset=publisher, null]
    }
    \map{
      \step[fieldsource=primaryclass, match=\regexp{\ANew Results\Z}, final]
      \step[fieldset=primaryclass, null]
    }
    \map{
      \step[fieldsource=chapter, match=\regexp{\ANew Results\Z}, final]
      \step[fieldset=chapter, null]
    }
    \map{
      \step[fieldsource=pages, match=\regexp{\A\d{4}\.\d{2}\.\d{2}}, final]
      \step[fieldset=pages, null]
    }
    \map{
      \step[fieldsource=number, match=\regexp{\AarXiv:}, final]
      \step[fieldset=number, null]
    }
    \map[overwrite]{
      \step[fieldsource=archiveprefix, matchi=\regexp{\Aarxiv\Z}, final]
      \step[fieldset=pubstate, fieldvalue={preprint}]
    }
    \map[overwrite]{
      \step[fieldsource=doi, match=\regexp{\A10\.(1101|64898)/}, final]
      \step[fieldset=pubstate, fieldvalue={preprint}]
    }
  }
}
\NewBibliographyString{preprint}
\DefineBibliographyStrings{english}{preprint = {preprint}}
\renewbibmacro{in:}{}
\AtEveryBibitem{%
  \clearfield{issn}%
  \clearfield{urlyear}\clearfield{urlmonth}\clearfield{urlday}%
  \ifboolexpr{ test {\iffieldundef{doi}} and test {\iffieldundef{eprint}} }
    {}%
    {\clearfield{url}}%
}

\newcommand{\rev}[1]{#1}

\theoremstyle{plain}
\newtheorem{theorem}{Theorem}[section]
\newtheorem{lemma}[theorem]{Lemma}
\newtheorem{proposition}[theorem]{Proposition}
\newtheorem{corollary}[theorem]{Corollary}

\theoremstyle{definition}

\newtheorem{assumption}[theorem]{Assumption}

\theoremstyle{remark}
\newtheorem{remark}[theorem]{Remark}

\newcommand{\R}{\mathbb{R}}
\newcommand{\diff}{\mathrm{d}}
\newcommand{\LL}{\mathcal{L}}

\newcommand{\AAop}{\mathcal{A}}

\definecolor{okabeBlue}{HTML}{0072B2}
\definecolor{okabeOrange}{HTML}{E69F00}
\definecolor{okabeGreen}{HTML}{009E73}
\definecolor{okabeVermillion}{HTML}{D55E00}
\definecolor{okabePurple}{HTML}{CC79A7}
\definecolor{okabeSky}{HTML}{56B4E9}
\definecolor{tableGray}{gray}{0.35}
\newcommand{\cmark}{\textcolor{okabeGreen}{\ensuremath{\checkmark}}}
\newcommand{\xmark}{\textcolor{okabeVermillion}{\ensuremath{\times}}}

\usepackage{formatting}
\hypersetup{
  pdftitle={Identifiability Limits of Physics-Informed Inference for Spatial Stochastic Dynamics from Static Snapshots},
  pdfauthor={Rujie (Rebecca) Gu, Ray Zirui Zhang, and Christopher E. Miles}
}

\title{Identifiability Limits of Physics-Informed Inference for Spatial Stochastic Dynamics from Static Snapshots}

\author{\authorblock{%
\authorentry{Rujie (Rebecca) Gu}{Department of Mathematics, University of California, Irvine}{rujieg@uci.edu}
\authorentry[0000-0002-0410-5905]{Ray Zirui Zhang}{Department of Mathematical Sciences, Worcester Polytechnic Institute}{rzhang8@wpi.edu}
\authorentry*[0000-0001-5494-403X]{Christopher E.\ Miles}{Department of Mathematics, University of Utah}{chris.miles@utah.edu}
}}

\date{}

\begin{document}

\maketitle

\begin{abstract}
Despite increasing scale and resolution, many biological measurements \rev{provide only a static snapshot of molecular positions rather than the dynamics that produced them}. By combining flexible representations with mechanistic constraints, physics-informed machine learning offers a promising route to inferring these dynamics from static snapshots.
Motivated by subcellular imaging of gene expression, we ask when a static spatial pattern of molecules can identify spatially varying diffusivity, creation, destruction, and boundary exchange, and how different inference schemes perform on the task.
A structural identifiability analysis shows that distributed sources are non-identifiable, whereas a point source such as a transcription site can restore identifiability. These limits are further shaped by seemingly innocuous modeling choices: the boundary conditions, the spatial regularity of the underlying dynamics, and even the stochastic calculus convention.
We then adapt several physics-informed schemes, differing in how they represent the solution and enforce the governing equations, and demonstrate effective inference from a single snapshot.
Physics-informed approaches can thus recover spatial heterogeneities of biological dynamics from static data, but their use should be accompanied and guided by careful identifiability analysis for meaningful interpretation of the results.
\end{abstract}

\section{Introduction}
\label{sec:introduction}

\subsection{Motivation}

Many experimental techniques for studying biological dynamics are destructive: they fix or lyse the sample and therefore provide only a static snapshot. This is the case for much single-molecule and single-cell imaging, which records the positions or quantities of molecules but not the processes that produced them. It is then natural to ask how much of those dynamics, and their spatial dependence, can be recovered from the snapshot alone.

One such example of static spatial data that encodes underlying dynamics of interest is imaging of subcellular RNA. Techniques such as single-molecule fluorescence in situ hybridization (smFISH) \cite{rajImagingIndividualMRNA2008} and multiplexed variants such as MERFISH \cite{chenSpatiallyResolvedHighly2015} and seqFISH \cite{engTranscriptomescaleSuperresolvedImaging2019} measure precise spatial coordinates $\{x_1, \dots, x_N\}$ of individual messenger RNA (mRNA) molecules in fixed cells. The missing dynamics that generate these spatial patterns encode transcription, degradation, boundary export, and transport through the nuclear interior \cite{novoselsky2026SubcellularMRNALocalization}. These dynamics are inherently spatially varying, as the nucleus is organized into chromatin domains of varying compaction, nuclear bodies such as speckles \cite{wuDynamicsRNALocalization2024}, and interchromatin channels, structures that shape how mRNA moves from transcription sites toward export and where processing steps such as splicing occur \cite{lecuyerGlobalAnalysisMRNA2007, buxbaumRightPlaceRight2015, dingConstitutiveSplicingEconomies2019}. The diffusion coefficient $D(x)$ is one such spatially varying quantity, summarizing how the surrounding chromatin organization shapes mRNA motion, so a map of it would offer a functional counterpart to structural measurements of chromatin architecture \cite{vargasMechanismMRNATransport2005}. The difficulty is that the relevant organization cannot be measured exhaustively in the same fixed cell, so these spatially varying quantities must be inferred from the observed positions rather than measured directly.

Physics-informed machine learning builds scientific knowledge, such as governing equations or conservation laws, directly into a model rather than learning from data alone \cite{karniadakisPhysicsinformedMachineLearning2021}. \rev{We use ``physics-informed'' for these methods, which fall within the broader class of hybrid models and are sometimes called ``biologically informed'' in biological applications \cite{lagergrenBiologicallyinformedNeuralNetworks2020}. These methods use mechanistic constraints to fit flexible neural representations to sparse and noisy data \cite{cuomoScientificMachineLearning2022}, with applications to stochastic particle systems \cite{chen2021SolvingInverseStochastic}, reaction-diffusion models \cite{raoEncodingPhysicsLearn2023, laveryPhysicsInformedNeuralNetworks2026}, and systems biology \cite{danekerSystemsBiologyIdentifiability2023}.} Their flexibility and mechanistic constraints make them natural candidates for recovering dynamics from static data, where the constraints compensate for the missing temporal information \cite{NEURIPS2025_af6c0a1b}. The practice and limits of this approach, however, remain underdeveloped. \rev{A method can return an estimate without establishing whether the snapshot uniquely determines the dynamics that produced it.}

We study this question for stochastic particles undergoing birth, diffusion, and first-order degradation, observed at steady state through a single point pattern. We build on spatial point-process models \cite{miles2024InferringStochasticRates, miles2025IncorporatingSpatialDiffusion} of nuclear mRNA in which the observations exactly follow a spatial Poisson process, with intensity $u(x)$ generated by the dynamics, so that a single snapshot jointly encodes the stochastic number and positions of the particles. \rev{These earlier studies \cite{miles2024InferringStochasticRates, miles2025IncorporatingSpatialDiffusion}} inferred reaction and transport parameters with the diffusion coefficient specified or spatially homogeneous. \rev{In the present work, we ask what a single snapshot can determine when the diffusivity $D(x)$, source $B(x)$, degradation $\gamma(x)$, and boundary exchange are not all known in advance (\cref{fig:cartoon}). In the analysis, we absorb spatial variation in $\gamma(x)$ into the effective transport and source terms.}

\begin{figure}[t]
    \centering
    \includegraphics[width=0.95\linewidth]{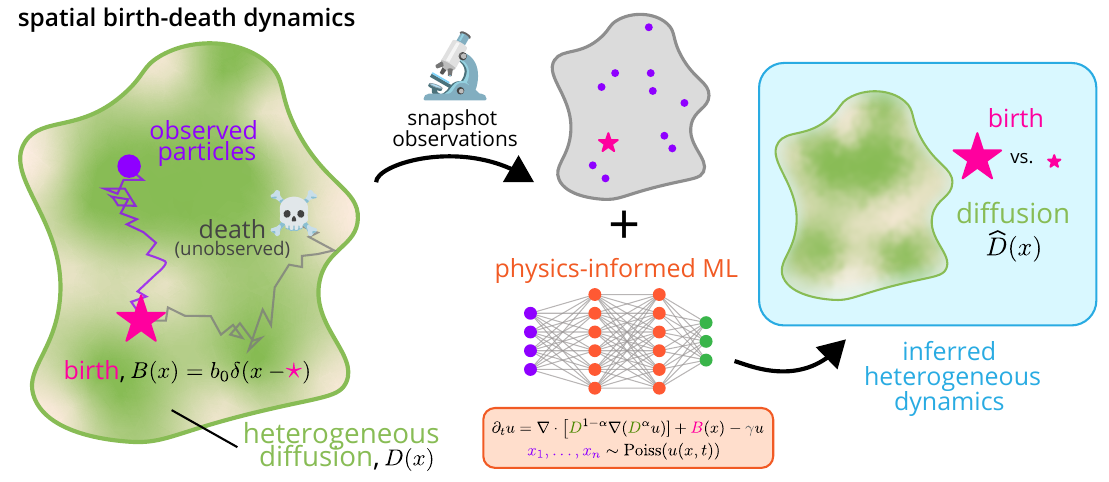}
    \caption{Schematic of the inverse problem. Particles are produced by a source $B(x)$, move with spatially varying diffusivity $D(x)$, and degrade. At steady state they are observed only through a static snapshot of their positions, which form a spatial Poisson process with intensity $u(x)$, random in both number and location. The inverse problem is to recover features of the underlying dynamics, such as $D(x)$ and \rev{the magnitude of a known source (shown as star size)}, from a single snapshot using physics-informed machine learning methods.}
    \label{fig:cartoon}
\end{figure}

\subsection{Prior work}

\paragraph{Trajectory-based inference}
When particles can be tracked over time, as in single-particle tracking experiments, a range of methods resolve spatial heterogeneity in the transport parameters  \cite{lin2025SuperresolvingParticleDiffusion}. Modern approaches span hidden Markov models \cite{slator2015DetectionDiffusionHeterogeneity}, Gaussian processes for $D(x)$ \cite{kumar2025DiffMAPGPContinuous2D}, and deep-learning inference of pointwise diffusivity \cite{requena2023InferringPointwiseDiffusion}. All of these fundamentally rely on observing displacements over known time intervals, and such trajectory data are not always available. Many high-throughput imaging modalities require fixation, or photobleaching limits track lengths. The question of what can be learned from \emph{snapshot data alone} has received comparatively little systematic attention.

\paragraph{Dynamics from static measurements}
More broadly, extracting dynamics from static or cross-sectional data is \rev{an active area of research} across quantitative biology. Many experimental modalities (single-cell sequencing, spatial transcriptomics, tissue-level proteomics) produce snapshots rather than time series, and a range of methods infer temporal quantities from them, such as rates \cite{somer2024TemporalTissueDynamics, zhang2025InferringStochasticDynamics}, velocities \cite{lamanno2018RNAVelocitySingle}, and developmental trajectories \cite{schiebingerOptimaltransportAnalysisSinglecell2019}. Alongside these methodological developments, complementary work has pursued understanding what such snapshots can and cannot determine \cite{cinquemani2018IdentifiabilityReconstructionBiochemical, guan2024IdentifyingDriftDiffusion, simpson2026WhenTrajectoriesMatter}. One emergent theme of these works, notably \cite{weinreb2018FundamentalLimitsDynamic}, is that some features of expression dynamics cannot be recovered from a static snapshot alone. However, that analysis, and most tools, are in expression-state space rather than physical quantities. The spatially resolved version, with physical transport parameters as the unknowns, remains largely open, and is the question we take up here.

\paragraph{PDE-constrained inverse problems}
The PDE inverse problems literature offers relevant mathematical tools but addresses a different data regime. Classical results, such as those for Calder\'on's problem in electrical impedance tomography \cite{uhlmannElectricalImpedanceTomography2009} and non-uniqueness in diffusion-based optical tomography \cite{arridge1998NonuniquenessDiffusionbasedOptical}, concern recovery of interior coefficients from boundary data, typically under Gaussian noise assumptions. Identifiability of PDE coefficients has been analyzed in other specific settings \cite{bachmayr2019IdentifiabilityDiffusionCoefficients, ciocanel2024ParameterIdentifiabilityPDE}, but recovering an entire field $D(x)$ from interior Poisson observations, whose randomness is generated by the dynamics rather than added as measurement noise, differs from both and has received less attention from either community.

\subsection{Contributions}

We make three contributions based on the heterogeneous model in \cref{sec:framework}:
\begin{itemize}
\item We prove identifiability conditions for steady-state birth-diffusion-degradation models with spatially heterogeneous coefficients, showing how recovery of $D(x)$ and the source depends on the boundary condition and on the stochastic interpretation (It\^o, Stratonovich, or Fickian). These cover the single-source cases with known boundaries and the non-identifiability of an unknown Robin boundary.
\item \rev{To demonstrate how well diffusivity can be recovered from a single snapshot in practice, we adapt one direct optimization method \cite{karnakov2023SolvingInverseProblems} and two neural methods \cite{raissiPhysicsinformedNeuralNetworks2019,zhangBiLOBilevelLocal2026} to the singular-source likelihood. In our tests, the spatial shape of $D(x)$ is recovered more reliably than its absolute scale.}
\item When a single snapshot is insufficient, additional measurements or richer structure can break the remaining ambiguity. We provide proof sketches for three further setups in which the parameters remain identifiable despite an unknown boundary permeability: multiple point sources with shared intensity, a downstream species sourced by an observed precursor, and two distinct steady states.
\end{itemize}

\section{Preliminaries}
\label{sec:framework}

\subsection{Observation model and steady-state equation}

Let $\Omega \subset \R^d$ denote the spatial domain (e.g., the cell nucleus), and let $X = \{x_1, \dots, x_N\}$ denote the observed particle positions \rev{in a single snapshot, so $N$ is the total number of particles observed}. Both $N$ and the positions $\{x_i\}$ are random.

For linear reaction networks involving first-order birth, death, and diffusion, the Poisson representation gives an exact reduction. When the probability distribution is written as a mixture of independent Poisson distributions across voxels, the system collapses to a deterministic intensity field $u(x)$ satisfying the reaction-diffusion PDE derived from the Chemical Master Equation, a correspondence formalized through spatio-temporal Cox processes \cite{schnoerrCoxProcessRepresentation2016}. At steady state ($\partial_t u = 0$), this intensity field satisfies the elliptic PDE:
\begin{equation}
\label{eq:steady-state}
\LL u(x) - \gamma u(x) + B(x) = 0 \quad \text{in } \Omega,
\end{equation}
subject to boundary conditions on $\partial \Omega$. \rev{Here $B(x)$ is the source, $\gamma$ is the degradation rate, and $\LL$ is the diffusion operator, which depends on the diffusivity $D(x)$ (for constant $D$, $\LL u=D\Delta u$) and is discussed in further detail in the next subsection.}

We view the observed particle locations $\{x_1, \dots, x_N\}$ as a spatial Poisson point process with intensity $u(x)$ \cite{mollerModernStatisticsSpatial2007}. For parameters $\theta$\rev{, comprising $D$, $B$, and any unknown boundary parameters,} the likelihood of \rev{the snapshot} $X$ is
\begin{equation}
\label{eq:likelihood}
L(\theta; X) = \left( \prod_{i=1}^{N} u(x_i; \theta) \right) \exp\left(-\int_{\Omega} u(x; \theta) \, \diff x\right).
\end{equation}
This likelihood incorporates both the total number of particles through the integral term and their precise spatial locations through the product term.

\begin{remark}[Identifiability reduces to uniqueness of $u(x)$]
\label{rem:mle-consistency}
We analyze structural identifiability \cite{villaverdeObservabilityStructuralIdentifiability2019} at the level of the intensity field $u(x)$. Distinct intensities define distinct Poisson point-process laws: if $u_1$ and $u_2$ differ on a set of positive measure, then \rev{the Kullback--Leibler divergence}
$D_{\mathrm{KL}}(u_1 \| u_2)=\int_\Omega [u_1 \log(u_1/u_2) - u_1 + u_2]\,\diff x>0$.
Parameter identifiability therefore reduces to injectivity of the PDE map $\theta \mapsto u$. \rev{The same reduction applies to any other observation model that determines $u$, so the results below do not depend on the Poisson form of the likelihood.}
\end{remark}

\subsection{Spatially varying diffusion}

The form of the diffusion operator $\mathcal{L}$ in heterogeneous media is not uniquely determined by the statement ``particles diffuse with coefficient $D(x)$.'' Different microscopic interpretations of stochastic calculus in spatially varying environments yield different macroscopic PDEs \cite{volpeEffectiveDriftsDynamical2016, pacheco-pozo2024LangevinEquationHeterogeneous}. To analyze identifiability systematically, we introduce a one-parameter family of diffusion operators:
\begin{equation}
\label{eq:alpha-operator}
\mathcal{L}_\alpha u \equiv \nabla \cdot \left[ D(x)^\alpha \nabla \left( D(x)^{1-\alpha} u(x) \right) \right],
\end{equation}
where $\alpha \in [0, 1]$ parameterizes the ``stochastic convention.'' Expanding the divergence term for each convention reveals the structural differences in the operators:
\begin{align*}
\LL_0 u &= \Delta(Du) = D \Delta u + 2 \nabla D \cdot \nabla u + u \Delta D, &&\!\!\! (\text{It\^o},\ \alpha=0) \\
\LL_{1/2} u &= \nabla \cdot \bigl(\sqrt{D}\,\nabla(\sqrt{D}\,u)\bigr) = D \Delta u + \tfrac{3}{2} \nabla D \cdot \nabla u + \tfrac{1}{2} u \Delta D, &&\!\!\! (\text{Stratonovich},\ \alpha=1/2) \\
\LL_1 u &= \nabla \cdot (D \nabla u) = D \Delta u + \nabla D \cdot \nabla u. &&\!\!\! (\text{Fickian},\ \alpha=1)
\end{align*}

The It\^o and Stratonovich expansions include a term proportional to $\Delta D$, a second derivative of the diffusion coefficient, whereas the Fickian expansion involves only first derivatives of $D$. This order difference is the source of the distinct identifiability behavior below.

From a physical perspective, the parameter $\alpha$ corresponds to distinct microscopic scenarios. Consider simulating a random walker in heterogeneous $D(x)$ where, at each timestep, we update position according to
\[
X_{n+1} = X_n + \sqrt{2 D(\star)\,\Delta t}\,\xi_n, \qquad \xi_n \sim N(0, 1).
\]
The convention specifies where the diffusion coefficient is evaluated. In the It\^o convention ($\alpha = 0$), $D$ is evaluated at the current position $X_n$, implying the particle does not anticipate the spatial field. The Stratonovich convention ($\alpha = 1/2$) evaluates $D$ at the midpoint, which preserves the standard chain rule. Fickian diffusion ($\alpha = 1$) corresponds to the macroscopic law where flux is proportional to the concentration gradient. The choice of stochastic calculus therefore encodes an assumption about the physical noise.

\subsection{Boundary conditions}

We impose boundary conditions through the physical flux associated with the chosen convention,
\begin{equation}
\label{eq:alpha-flux}
J_\alpha[u,D] \equiv -D(x)^\alpha \nabla\!\left(D(x)^{1-\alpha}u(x)\right).
\end{equation}
Robin boundary conditions then prescribe outward flux proportional to local concentration:
\begin{equation}
\label{eq:robin-bc}
D(x)^\alpha \frac{\partial}{\partial n}\!\left(D(x)^{1-\alpha}u(x)\right) + \kappa(x)u(x) = 0
\quad \text{on } \partial \Omega,
\end{equation}
where $\kappa(x)$ is the boundary permeability. In the one-dimensional setting used below, we write the endpoint permeabilities as $\kappa_0=\kappa(0)$ and $\kappa_L=\kappa(L)$; when a single $\kappa$ is written, we mean the symmetric case $\kappa_0=\kappa_L=\kappa$. This reduces to $D\partial_n u+\kappa u=0$ for Fickian diffusion and to $\partial_n(Du)+\kappa u=0$ for It\^o diffusion.

The limiting cases recover familiar boundary conditions:
\begin{itemize}
    \item \textbf{Dirichlet (absorbing):} $\kappa \to \infty$, so $u = 0$. Particles are removed upon reaching the boundary.
    \item \textbf{Neumann (reflecting):} $\kappa = 0$, so $J_\alpha \cdot \mathbf{n} = 0$. Particles are reflected at the boundary.
    \item \textbf{Robin (semi-permeable):} $0 < \kappa < \infty$, interpolating between the absorbing and reflecting limits. For nuclear export, this models finite permeability through nuclear pore complexes.
\end{itemize}

\begin{remark}[Nondimensionalization]
\label{rem:nondim}
Throughout the proofs we work on a general interval $[0,L]$ and take $\gamma>0$ to be constant. Rescaling space by $x\mapsto x/L$ and time by $t\mapsto\gamma t$ removes the domain length and degradation timescale, so these absolute scales do not affect the identifiability arguments. The unknowns are $D(x)$, $B(x)$, and, when present, the boundary permeability $\kappa$. Holding $\gamma$ constant folds any spatial variation of the degradation rate into $D$ and $B$. This is only approximate, since dividing by a variable $\gamma(x)$ does not simply rescale $D$: $\frac{1}{\gamma(x)}(D(x)u(x))'' \neq \left(\frac{D(x)}{\gamma(x)}u(x)\right)''$ unless $\gamma$ is constant. We expect the qualitative interpretation of the results below to persist for a general $\gamma(x)$.
\end{remark}

\subsection{The ambiguity operator}

Suppose two parameter sets $(D_1, B_1)$ and $(D_2, B_2)$, with positive diffusivities and nonnegative sources, produce the \emph{same} steady-state solution $u(x)$. Define the perturbations
\begin{equation}
\delta D(x) = D_2(x) - D_1(x), \quad \delta B(x) = B_2(x) - B_1(x).
\end{equation}

Subtracting their steady-state equations gives the \emph{ambiguity equation}
\begin{equation}
\label{eq:ambiguity}
\AAop_{u,\alpha}[\delta D] + \delta B = 0.
\end{equation}
Rearranging the steady-state equation, the profile $u$ fixes $\LL_\alpha u = \gamma u - B$ at every interior point. It does not fix $D$ by itself.

The operator $\AAop_{u,\alpha}$ is obtained by expanding the diffusion operator at fixed $u$. Although $\LL_\alpha$ contains powers of $D$, those powers cancel:
\begin{align}
\LL_\alpha u
&= \nabla \cdot \left[D^\alpha \nabla\left(D^{1-\alpha}u\right)\right] \\
&= D \Delta u + (2-\alpha)\nabla D \cdot \nabla u
   + (1-\alpha)u\,\Delta D. \label{eq:expanded-alpha}
\end{align}
Thus $\LL_\alpha u$ is linear in $D$ once $u$ is fixed. Replacing $D$ by $D+\delta D$ in \eqref{eq:expanded-alpha} and subtracting gives the exact ambiguity operator
\begin{equation}
\label{eq:ambiguity-expanded}
\AAop_{u,\alpha}[\delta D]
= \delta D\,\Delta u
  + (2-\alpha)\nabla(\delta D)\cdot\nabla u
  + (1-\alpha)u\,\Delta(\delta D).
\end{equation}
The ambiguity equation is therefore not a linearization or first-order approximation; the identifiability results below are global within the stated model class. Although $\gamma$ cancels from the ambiguity equation, we assume $\gamma>0$ throughout, and several of the results below rely on it.

Two special cases drive the rest of the paper. For Fickian diffusion $(\alpha=1)$, $\AAop_{u,1}[\delta D] = \nabla\cdot(\delta D\,\nabla u)$, which is first-order in $\delta D$. For It\^o diffusion $(\alpha=0)$, $\AAop_{u,0}[\delta D] = \Delta(\delta D\,u)$, which is second-order in $\delta D$. The convention $\alpha$ therefore determines which composite of $D$ and $u$ a snapshot constrains.

By \cref{rem:mle-consistency}, the parameter set $(D, B)$ is structurally identifiable from observation of $u(x)$ if the only solution to the ambiguity equation, subject to appropriate boundary conditions and regularity constraints, is the trivial solution $\delta D \equiv 0$, $\delta B \equiv 0$.

\paragraph{Proof strategy}
Every identifiability proof \rev{in this work} compares two parameter sets that produce the same profile $u$ and studies their difference $\delta D$. The useful variable depends on the stochastic convention,
\[
\text{It\^o: } f=\delta D\,u, \qquad \text{Fickian: } g=\delta D\,u'.
\]
The ambiguity equation is second order in $f$ and first order in $g$. The point source and boundary conditions then act on these variables through corners, jumps, values, or slopes. Identifiability holds when these constraints force $\delta D\equiv0$ and $\delta B\equiv0$.

\section{Diffuse sources are non-identifiable}
\label{sec:diffuse}

For diffuse sources, a source perturbation can be offset by a diffusivity perturbation while leaving the same steady-state density.

\begin{theorem}[Diffuse source non-identifiability]
\label{thm:diffuse-source}
If the birth rate $B(x)$ is a smooth, diffuse function (rather than a singular point source), the pair $(D(x), B(x))$ is non-identifiable for any choice of $\alpha$.
\end{theorem}

\begin{proof}
Compare $(D, B)$ with a perturbed pair $(D + \delta D,\, B + \delta B)$ at the fixed profile $u$. By \eqref{eq:ambiguity-expanded} the two produce the same $u$ exactly when $\AAop_{u,\alpha}[\delta D] + \delta B = 0$, which for It\^o $(\alpha=0)$ reads $(\delta D\, u)'' = -\delta B$. We build a nontrivial solution by choosing the diffusivity perturbation first. Take any nonzero smooth $\delta D$ supported in the interior, away from the boundaries, and small enough that $D + \delta D > 0$ and $B + \delta B \ge 0$, and set the matching source perturbation to $\delta B := -\,\AAop_{u,\alpha}[\delta D]$.
This satisfies the ambiguity equation by construction, so $(D + \delta D,\, B + \delta B)$ reproduces the same profile $u$.

The perturbed pair is admissible. Since $\delta D$ vanishes near the boundaries, so does $\delta B$, and the boundary conditions on $u$ are left intact. For reflecting boundaries the source must also satisfy the mass-balance compatibility $\int_0^L \delta B = 0$; this holds automatically, because $\AAop_{u,\alpha}[\delta D]$ is the divergence of a flux supported in the interior, so
\[
\int_0^L \delta B \,\diff x = -\int_0^L \AAop_{u,\alpha}[\delta D] \,\diff x = 0.
\]
The perturbation $\delta D$ ranges over an infinite-dimensional space, so this confounding is present for every $\alpha$.
\end{proof}

An elevated concentration can reflect local production or locally slow diffusion, and only singular sources introduce the derivative jumps that separate the two. As an explicit example, consider $[0, L]$ with Dirichlet boundaries and true parameters $(D_1, B_1)$ producing $u$. A constant source offset $\delta B(x) = C$ is compensated exactly: solving $f'' = -\delta B$ with $f(0) = f(L) = 0$ gives $f(x) = \tfrac{C}{2}\,x(L - x)$, and setting $\delta D = f/u$ makes $(D_2, B_2) = (D_1 + \delta D, B_1 + \delta B)$ reproduce the same solution $u$. \Cref{fig:Ito_Dirichlet} plots this construction for a Gaussian base source with offset $C = 1$; re-solving the PDE from $(D_2, B_2)$ returns the original density to numerical precision.

\begin{figure}[htbp]
    \centering
    \includegraphics[width=\linewidth]{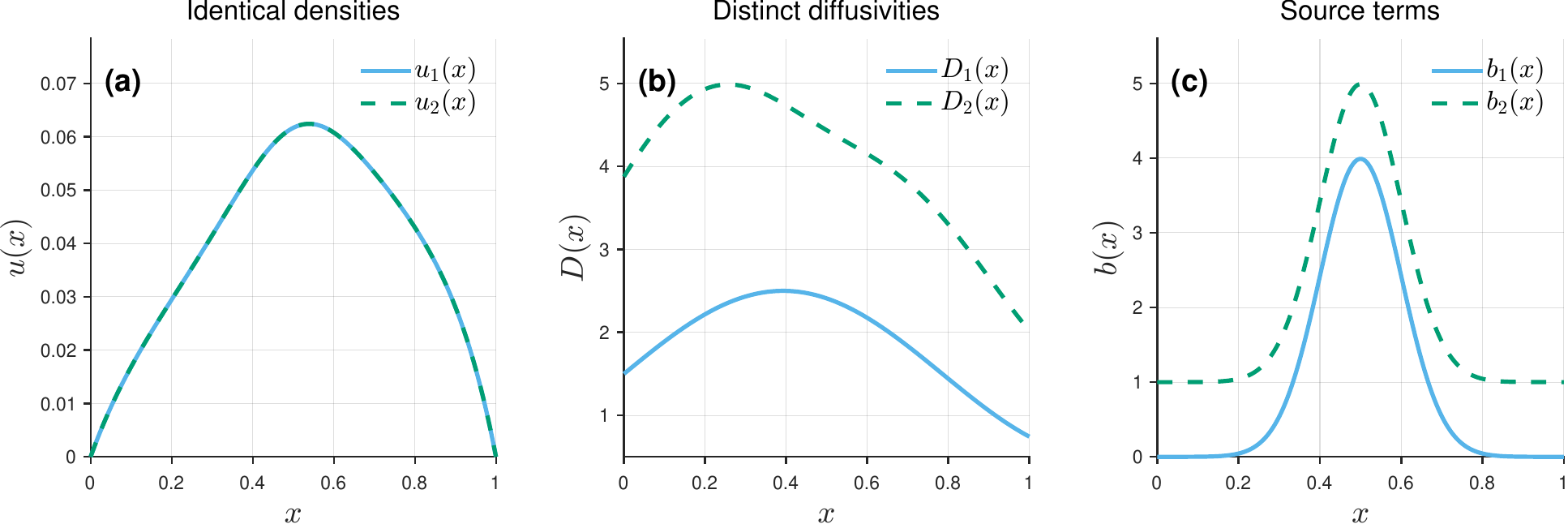}
    \caption{Same-density diffuse-source alternative under It\^o dynamics with Dirichlet boundaries. Here $\gamma = 10$, $D_1(x) = 1.5 + \sin(4x)$, and $b_1(x) = \tfrac{1}{\sigma\sqrt{2\pi}}\exp\!\big(-(x-0.5)^2/2\sigma^2\big)$ with $\sigma = 0.1$; the source perturbation is $b_2=b_1+1$. The compensating $D_2 = D_1+\delta D$ follows the construction in \cref{thm:diffuse-source}. Panels show (a)~the common density, (b)~\rev{the diffusivities $D_1$ and $D_2$, which give the same density}, and (c)~the alternative sources $b_1$ and $b_2$.
}
    \label{fig:Ito_Dirichlet}
\end{figure}

\section{Single point-source identifiability}
\label{sec:single-species}

\subsection{Setup}
\label{sec:point-source}

\Cref{thm:diffuse-source} shows that, for diffuse sources, a smooth source perturbation can be offset by a smooth diffusivity perturbation. A point source instead creates a localized jump condition in the density and is the natural model for transcription at a gene locus. We therefore take
\[
B(x)=b_0\,\delta(x-z)
\]
with known position $z$ and unknown strength $b_0$. In single-molecule imaging, the active transcription site appears as a distinct bright spot, where nascent transcripts accumulate at the gene locus, and can be localized directly \cite{dingConstitutiveSplicingEconomies2019, tsanovSmiFISHFISHquantFlexible2016}. The density remains continuous, but when $D$ is $C^1$ at the source, the jump condition gives
\[
[u'](z)=-\frac{b_0}{D(z)}.
\]
Here $[w](z) := w(z^+) - w(z^-)$ denotes the jump of a quantity $w$ across $z$, the limit from the right minus the limit from the left.

The relevant distinction is whether $D$ may kink at the source. If continuous kinked diffusivities are allowed, the source signature can be moved from $b_0$ into a kink of $D$, producing an observationally equivalent parameter set. Requiring $D$ to be $C^1$ at $z$ assigns the singularity to the source rather than to the medium and excludes this reallocation.

\begin{assumption}[The singularity belongs to the source]
\label{ass:regularity}
At each source location $z_i$, the diffusion coefficient $D(x)$ is $C^1$ (continuously differentiable).
\end{assumption}

The diffusivity may still jump at the boundaries of nuclear condensates or chromatin domains, provided it is smooth at the gene locus itself; a transcription site is rarely located exactly at such a boundary. \rev{If one were, for example at the edge of a transcriptional condensate \cite{choMediatorRNAPolymerase2018}, a kink or jump of $D$ at the source could not be excluded.}

We ask what a single snapshot determines about $(D,b_0)$ when the boundary type, absorbing (Dirichlet) or reflecting (Neumann), is known\rev{, for example from the biological mechanism}. A partially permeable boundary\rev{, whose permeability $\kappa$ is a further unknown, is the more general case and} is treated separately in \cref{sec:extended-structure}. In the point-source problem below, the locations are fixed but the strengths may vary, producing jump conditions proportional to the perturbations $\delta b_i$.

\subsection{Perturbation lemmas}

After two parameter sets are subtracted, the raw unknown $\delta D$ is not the most convenient quantity to track: the closed equation is instead for $f=\delta D\,u$ in the It\^o case and for the flux perturbation $g=\delta D\,u'$ in the Fickian case. The next two lemmas record what the point source and the boundary say about these two quantities. Throughout, $u>0$ in the interior and $D>0$; the It\^o lemma also uses that $\delta D$ is $C^1$ at each source (\cref{ass:regularity}).

\begin{lemma}[It\^o source jump]
\label{lem:source-jump}
For It\^o diffusion ($\alpha = 0$) with point sources of nonzero strengths $b_i$ at distinct interior points $z_i$, suppose two parameter sets produce the same $u$ and write $f = \delta D\, u$ for the difference. Then
\begin{enumerate}
\item[(i)] $f'' = -\sum_i \delta b_i\,\delta(x - z_i)$ in the interior, so $f$ is continuous and piecewise linear, with a corner $[f'](z_i) = -\delta b_i$ at each source;
\item[(ii)] at each source,
\[
[u'](z_i) = -\frac{b_i}{D(z_i)}, \qquad
\delta D(z_i) = \delta b_i\,\frac{D(z_i)}{b_i};
\]
\item[(iii)] the inherited boundary conditions are: Dirichlet gives $f = 0$ at each endpoint, while reflecting (Neumann) boundaries give $f' = 0$ at each endpoint. Reflecting boundaries also yield the mass balance
\[
\sum_i b_i = \gamma\int_0^L u\,\diff x,
\]
which fixes the total source strength.
\end{enumerate}
\end{lemma}

\begin{proof}
For (i), the It\^o ambiguity equation is $(\delta D\, u)'' = -\sum_i \delta b_i\,\delta(x-z_i)$. Thus $f'' = 0$ away from the sources, so $f$ is piecewise linear, and integrating across $z_i$ gives the corner condition $[f'](z_i) = -\delta b_i$.

For (ii), integrate the unperturbed steady state $(Du)'' = \gamma u - \sum_i b_i\,\delta(x-z_i)$ across $z_i$. This gives $[(Du)'](z_i) = -b_i$. Since $D$ is $C^1$ at the source under \cref{ass:regularity} and $u$ is continuous, $[(Du)'](z_i)=D(z_i)[u'](z_i)$, so $[u'](z_i) = -b_i/D(z_i)$. Expanding $[f'](z_i) = [(\delta D)'\,u](z_i) + [\delta D\, u'](z_i)$, the first jump vanishes because $\delta D$ is $C^1$ at the source, so $-\delta b_i = [f'](z_i) = \delta D(z_i)\,[u'](z_i)$, which gives the stated value of $\delta D(z_i)$.

For (iii), subtract the boundary conditions of the two parameter sets. Under Dirichlet conditions $u = 0$ at each endpoint, so $f = \delta D\, u = 0$ there. Under reflecting conditions $(Du)' = 0$ for each set, so their difference gives $f' = (\delta D\, u)' = 0$ at each endpoint. Integrating the reflecting steady-state equation over $[0,L]$ removes the flux term and leaves the mass balance.
\end{proof}

\begin{lemma}[Fickian flux perturbation]
\label{lem:flux-gauge}
For Fickian diffusion ($\alpha = 1$) with $u'$ not identically zero, the flux perturbation $g := \delta D\, u'$ obeys $g' = -\delta B$, so
\[
g(x) = C_1 - \int_0^x \delta B(\xi)\,\diff\xi.
\]
In particular $g$ is constant when the source perturbation vanishes ($\delta B = 0$), and piecewise constant, $g(x) = C_1 - \sum_i \delta b_i\, H(x - z_i)$, where $H$ is the Heaviside step function, for point sources of perturbed strengths $\delta b_i$. At a reflecting boundary, the inherited flux condition forces $g$ to vanish there.
\end{lemma}

\begin{proof}
Integrating the Fickian ambiguity equation $(\delta D\, u')' = -\delta B$ once gives $g$ as stated. At a reflecting boundary, $Du' = 0$ for each parameter set; since $D > 0$, this forces $u' = 0$, and hence $g = \delta D\, u' = 0$ there.
\end{proof}

\begin{remark}[Stratonovich and intermediate conventions]
\label{rem:intermediate-conventions}
Although we focus the detailed proofs on It\^o diffusion, the same second-order mechanism extends to Stratonovich ($\alpha = 1/2$) and other conventions with $\alpha < 1$, since the $\Delta D$ term appears in the expansion of $\LL_\alpha$. Only the Fickian case ($\alpha = 1$) reduces to a first-order operator, so the \rev{flux perturbation of \cref{lem:flux-gauge}} is specific to $\alpha = 1$.
\end{remark}

\subsection{It\^o identifiability for a point source}

For the It\^o convention, the point source and the regularity assumption together suffice for identifiability, once the boundary \rev{condition} is known.

\begin{theorem}[It\^o point-source identifiability]
\label{thm:ito-point-source}
Consider the It\^o model ($\alpha = 0$) with a single point source $B(x) = b_0\,\delta(x - z)$ of known location $z$ and unknown strength $b_0$, decay $\gamma > 0$, and known boundary conditions (Dirichlet or Neumann). Within the admissible class in which $D$ is $C^1$ at the source (\cref{ass:regularity}), the pair $(D(x), b_0)$ is identifiable from the steady-state profile $u(x)$.
\end{theorem}

The Dirichlet Green's function on $[0, L]$, satisfying $-G''(x, z) = \delta(x - z)$ with $G(0, z) = G(L, z) = 0$, is
\begin{equation}
G(x, z) = \begin{cases} \frac{x(L-z)}{L} & x \leq z \\ \frac{z(L-x)}{L} & x > z, \end{cases} \qquad G(z, z) = \frac{z(L-z)}{L}.
\end{equation}

\begin{proof}
By \cref{lem:source-jump}, the auxiliary $f = \delta D\, u$ satisfies $f'' = -\delta b_0\,\delta(x - z)$, and the source corner (part~(ii)) fixes
\begin{equation}
\label{eq:constraint2}
\delta D(z) = -\frac{\delta b_0}{[u'](z)}, \qquad [u'](z) = -\frac{b_0}{D(z)},
\end{equation}
This quantity is nonzero for a nontrivial source. We treat the two boundary conditions separately. In each case, it remains to show that the boundary constraints and the source jump force $\delta b_0 = 0$.

\medskip\noindent\textbf{Dirichlet.}
By \cref{lem:source-jump}(iii), $f(0) = f(L) = 0$, and together with the delta source this determines $f$ completely as $f(x) = \delta b_0\, G(x, z)$, with $G$ the Dirichlet Green's function introduced above. Evaluating at $x = z$,
\begin{equation}
\label{eq:constraint1}
\delta D(z) \cdot u(z) = \delta b_0 \cdot G(z, z).
\end{equation}
The two constraints \eqref{eq:constraint1} and \eqref{eq:constraint2} can be written as the homogeneous system
\[
\begin{pmatrix}
u(z) & -G(z,z) \\
[u'](z) & 1
\end{pmatrix}
\begin{pmatrix}
\delta D(z) \\
\delta b_0
\end{pmatrix}
=0.
\]
A nonzero perturbation $(\delta D(z), \delta b_0)$ therefore requires the determinant to vanish:
\begin{equation}
u(z) + G(z,z) \cdot [u'](z) = 0.
\end{equation}
This determinant can be evaluated exactly. Writing $V = Du$ for the unperturbed product, the steady-state equation is $V'' = \gamma u - b_0\,\delta(x-z)$ with $V(0) = V(L) = 0$, so that
\begin{equation}
V(x) = D(x)u(x) = b_0\,G(x,z) - \gamma \int_0^L G(x,y)\,u(y)\,\diff y.
\end{equation}
Evaluating at $x = z$ and using $[u'](z) = -b_0/D(z)$,
\begin{equation}
u(z) + G(z,z)\,[u'](z) = \frac{1}{D(z)}\Big[ D(z)u(z) - b_0\,G(z,z) \Big] = -\frac{\gamma}{D(z)} \int_0^L G(x,z)\,u(x)\,\diff x,
\end{equation}
which is strictly negative whenever $\gamma > 0$, since $G > 0$ and $u > 0$ in the interior. The determinant is therefore nonzero, so $\delta b_0 = 0$; no genericity caveat is needed once $\gamma > 0$. With $\delta b_0 = 0$, the equation $f'' = 0$ with $f(0) = f(L) = 0$ forces $f \equiv 0$, and since $u > 0$ in the interior, $\delta D = f/u \equiv 0$.

\medskip\noindent\textbf{Neumann.}
For reflecting boundaries, integrate the steady-state PDE over $[0, L]$; the zero-flux conditions $\partial_x(Du) = 0$ at both ends give the global mass balance $b_0 = \gamma \int_0^L u\,\diff x$, which is fixed by the observed profile $u$, so $\delta b_0 = 0$ immediately. The ambiguity equation is then $f'' = 0$ with the perturbed flux $f' = (\delta D\,u)'$ vanishing at both endpoints, forcing $f \equiv C$ constant, hence $\delta D = C/u$. Differentiating, $\delta D' = -C u'/u^2$ inherits the jump in $u'$ at the source:
\begin{equation}
[\delta D'](z) = -C\,\frac{[u'](z)}{u(z)^2} = \frac{C\,b_0}{D(z)\,u(z)^2}.
\end{equation}
For $C \neq 0$ this is a genuine slope discontinuity in $\delta D$ at $z$, which \cref{ass:regularity} forbids; therefore $C = 0$ and $\delta D \equiv 0$. In both cases the system is structurally identifiable.
\end{proof}

\begin{remark}[Explicit same-density alternatives]
\label{rem:near-null}
The source-regularity assumption excludes the continuous same-density alternatives. Without it, each boundary leaves one explicitly. Under Dirichlet boundaries, change the source strength by a nonzero $\delta b_0$. Part~(iii) of \cref{lem:source-jump} gives $f = \delta b_0\, G(\cdot, z)$ with $G$ the Dirichlet Green's function, so
\[
\delta D = \frac{f}{u} = \frac{\delta b_0\, G(\cdot, z)}{u}, \qquad b_0^{(2)} = b_0^{(1)} + \delta b_0.
\]
Under reflecting boundaries, the mass balance fixes $\delta b_0 = 0$, leaving the homogeneous mode $f \equiv C$, so
\[
\delta D = \frac{C}{u}.
\]
In each case $(D_1 + \delta D,\, b_0^{(1)} + \delta b_0)$ reproduces the same density $u$, but $\delta D$ kinks at the source because $u'$ jumps there.
\end{remark}

Using the constructions in \cref{rem:near-null}, \cref{fig:Ito_Dirichlet_pointSource} plots the Dirichlet alternative and \cref{fig:Ito_Neumann_pointSource} the reflecting one. In both cases the density is reproduced exactly, while $\delta D$ is continuous but kinked at the source, a $C^0$ rather than $C^1$ perturbation.

\begin{figure}[tb]
    \centering
    \includegraphics[width=\linewidth]{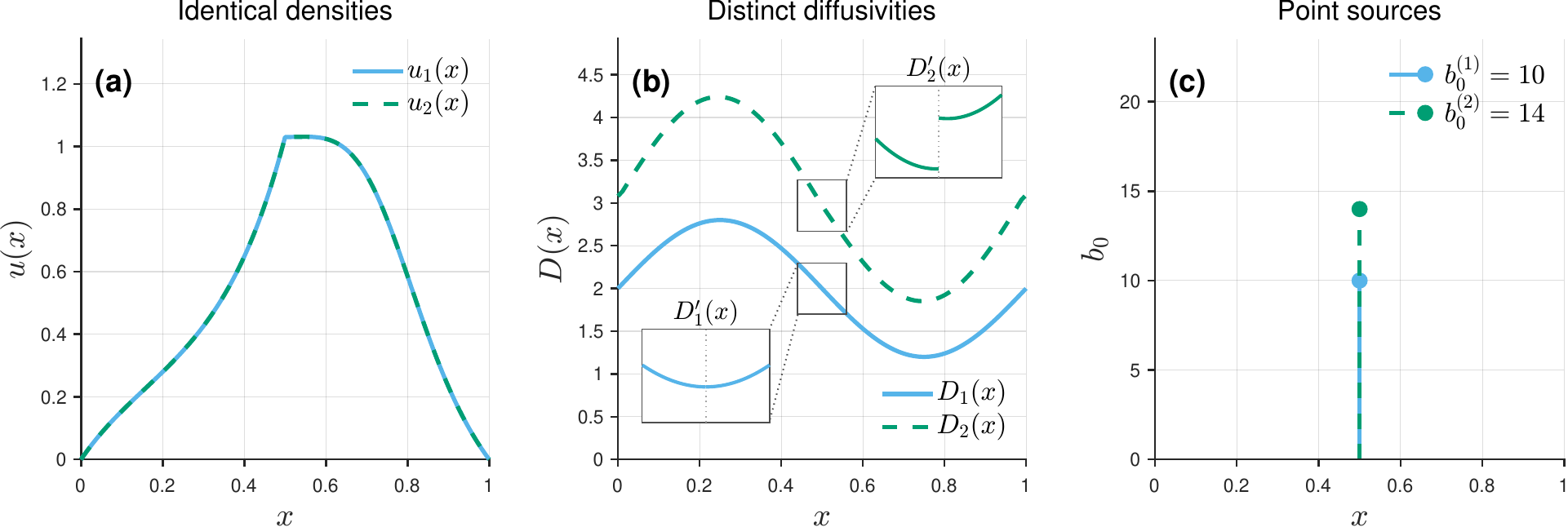}
    \caption{Same-density It\^o alternative under absorbing Dirichlet boundaries. Here $D_1(x) = 2 + 0.8\sin(2\pi x)$, $z = 0.5$, $\gamma = 5$, $b_0^{(1)} = 10$, and $b_0^{(2)} = 14$; the alternative $D_2 = D_1+\delta D$ follows the construction in \cref{rem:near-null}. Panels show (a)~the common density, (b)~\rev{the diffusivities $D_1$ and $D_2$, which give the same density}, with $\delta D$ continuous but kinked at $z$ \rev{(the insets show the derivatives over the boxed regions: $D_1'$ is continuous at $z$, while $D_2'$ jumps there)}, and (c)~the alternative source strengths.}
    \label{fig:Ito_Dirichlet_pointSource}
\end{figure}

\begin{figure}[tb]
    \centering
    \includegraphics[width=\linewidth]{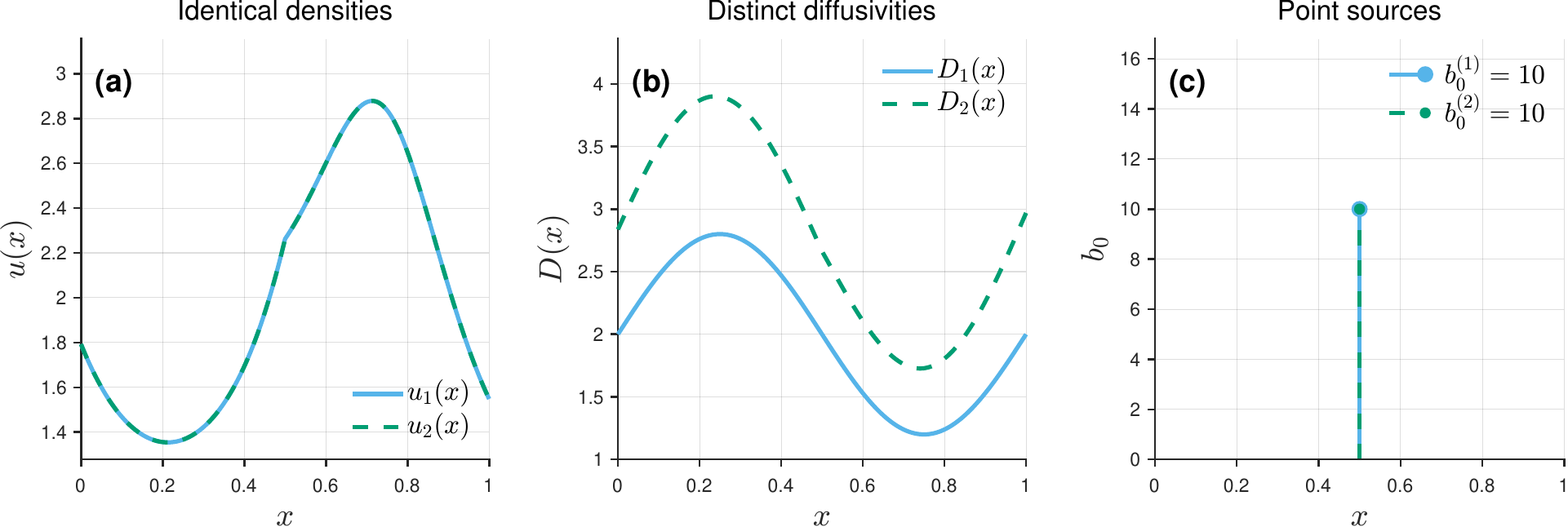}
    \caption{Same-density It\^o alternative under reflecting (Neumann) boundaries. Here $D_1(x) = 2 + 0.8\sin(2\pi x)$, $b_0 = 10$, $z = 0.5$, and $\gamma = 5$; the alternative $D_2 = D_1+\delta D$ uses the homogeneous mode $C=1.5$ from \cref{rem:near-null}. Panels show (a)~the common density, (b)~\rev{the diffusivities $D_1$ and $D_2$, which give the same density}, with $\delta D$ continuous but kinked at $z$, and (c)~the unchanged source strength.}
    \label{fig:Ito_Neumann_pointSource}
\end{figure}

\subsection{The Fickian flux ambiguity}

In the Fickian case the closed perturbation equation is first-order, for the flux perturbation $g=\delta D\,u'$, rather than second-order in $f=\delta D\,u$. A point source again produces a continuous kinked alternative unless the boundary condition fixes the flux directly.

\begin{theorem}[Fickian point-source identifiability]
\label{thm:fickian-nonident}
Consider the Fickian model ($\alpha=1$) with a single point source $B(x)=b_0\,\delta(x-z)$ of known location $z$, unknown strength $b_0$, and decay $\gamma>0$. Assume $u'\neq 0$ on each open interval between $z$ and the endpoints, and $u'(z^\pm)\neq0$ at the source. If the boundary condition does not fix the Fickian flux perturbation $g=\delta D\,u'$ (as under Dirichlet boundaries), the continuous admissible class contains a one-parameter family of same-density alternatives. Within the class where $D$ is $C^1$ at the source (\cref{ass:regularity}), the pair $(D,b_0)$ is identifiable.
\end{theorem}

\begin{proof}
Suppose another admissible pair produces the same profile $u$, and let $\delta b_0$ be its source-strength difference. By \cref{lem:flux-gauge} the flux perturbation is piecewise constant,
\[
g(x):=\delta D(x)\,u'(x)=C_1-\delta b_0\, H(x-z),
\]
so $\delta D=g/u'$ wherever $u'\neq0$, and this reproduces $u$ for any choice of $C_1$ and $\delta b_0$. It remains to see when such a $\delta D$ can be continuous and $C^1$ at the source.

The point source makes $u'$ jump, $[u'](z)=-b_0/D(z)\neq0$, and $u'(z^\pm)\neq0$ by \rev{the assumption of the theorem}, so continuity of $\delta D=g/u'$ at $z$,
\[
\frac{C_1}{u'(z^-)}=\frac{C_1-\delta b_0}{u'(z^+)},
\]
fixes $C_1=\delta b_0\,u'(z^-)D(z)/b_0$, and the common one-sided value is $d_0:=\delta D(z)=C_1/u'(z^-)=\delta b_0\,D(z)/b_0$. A nontrivial alternative has $\delta b_0\neq0$ (if $\delta b_0=0$ then $C_1=0$ and $\delta D\equiv0$), hence $d_0\neq0$.

Because $g$ is constant on each side of $z$, differentiating $\delta D=g/u'$ gives $\delta D'=-\delta D\,u''/u'$, so
\[
[\delta D'](z)=-d_0\left[\frac{u''}{u'}\right](z).
\]
Away from the source the steady-state equation $Du''+D'u'=\gamma u$ gives $u''/u'=\gamma u/(Du')-D'/D$, and since $D$ is \rev{positive and} $C^1$ at $z$\rev{, the quotient $D'/D$ is continuous there, so this} term has no jump\rev{,} leaving
\[
\left[\frac{u''}{u'}\right](z)=\frac{\gamma u(z)}{D(z)}\left[\frac{1}{u'}\right](z)=\frac{\gamma\,b_0\,u(z)}{D(z)^2\,u'(z^+)u'(z^-)},
\]
using $[1/u'](z)=\bigl(u'(z^-)-u'(z^+)\bigr)/\bigl(u'(z^+)u'(z^-)\bigr)=(b_0/D(z))/(u'(z^+)u'(z^-))$. Hence
\[
[\delta D'](z)=-d_0\,\frac{\gamma\,b_0\,u(z)}{D(z)^2\,u'(z^+)u'(z^-)},
\]
which is nonzero whenever $\gamma>0$. Every nontrivial continuous alternative thus has a genuine kink at the source, so \cref{ass:regularity} excludes it and forces $\delta D\equiv0$, $\delta b_0=0$.
\end{proof}

\Cref{fig:Fickian_Dirichlet} shows the construction for the Dirichlet case with a point source at $z$. Because $u'$ jumps there, keeping $\delta D = (C_1 - \delta b_0 H)/u'$ continuous forces the source strength to change too ($b_0^{(1)} \neq b_0^{(2)}$); the resulting $(D_2, b_0^{(2)})$ reproduces the same monotone density, with $\delta D$ kinked at $z$.

\begin{figure}[tb]
    \centering

    \includegraphics[width=\linewidth]{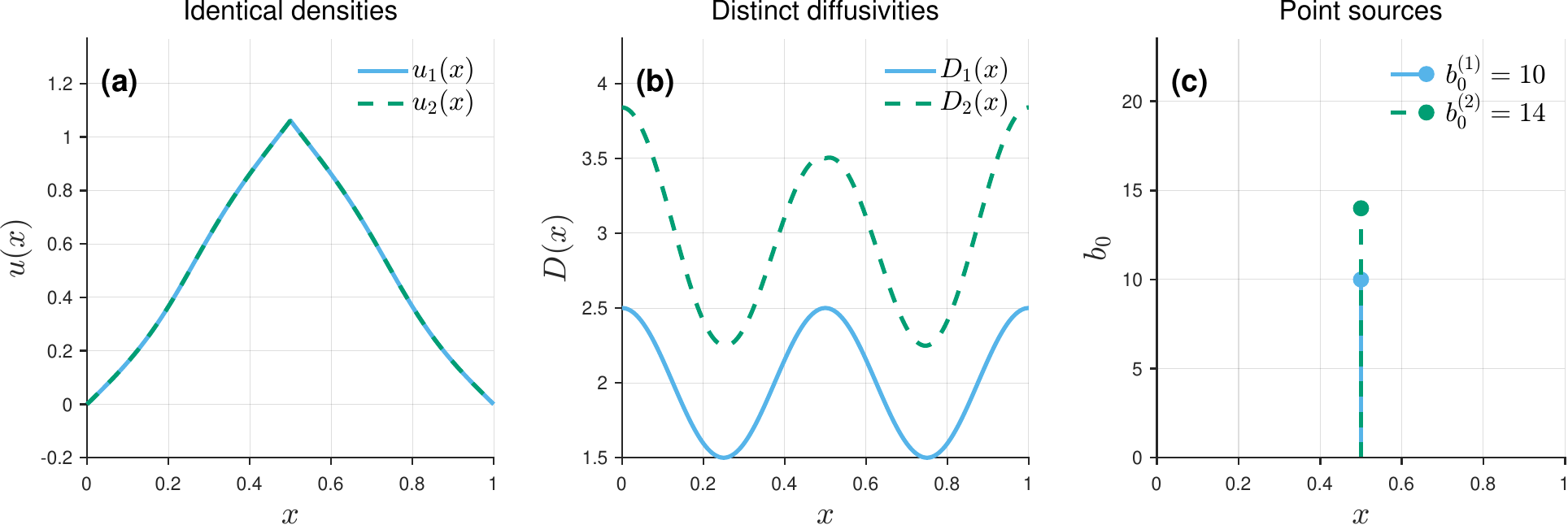}
    \caption{Same-density Fickian alternative under absorbing Dirichlet boundaries. Here $D_1(x) = 2 + 0.5\cos(4\pi x)$, $z = 0.5$, $\gamma = 5$, $b_0^{(1)}=10$, and $b_0^{(2)}=14$; the alternative $D_2 = D_1+\delta D$ follows the construction in \cref{thm:fickian-nonident}. Panels show (a)~the common density, with the source kink at $z$, (b)~\rev{the diffusivities $D_1$ and $D_2$, which give the same density}, with $\delta D$ continuous but kinked at $z$, and (c)~the alternative source strengths.
}
    \label{fig:Fickian_Dirichlet}
\end{figure}

\subsection{Reflecting boundaries in the Fickian case}

There is a single configuration that needs no regularity assumption at all. Reflecting (Neumann) boundaries force $u'=0$ at the boundary, and the Fickian ambiguity is exactly the flux perturbation $\delta D\,u'$.

\begin{corollary}[Reflecting boundaries identify the Fickian point-source model]
\label{cor:neumann-fickian}
Consider the Fickian model ($\alpha=1$) with a single point source $B(x)=b_0\,\delta(x-z)$, known source location $z$, unknown strength $b_0$, and reflecting (Neumann) boundary conditions at both endpoints. Then the pair $(D,b_0)$ is identifiable from the steady-state profile $u$.
\end{corollary}

\begin{proof}
Let another pair $(\widetilde D,\widetilde b_0)$ produce the same profile $u$, and set $\delta D=\widetilde D-D$ and $\delta b_0=\widetilde b_0-b_0$. By \cref{lem:flux-gauge}, the Fickian flux perturbation has the form
\[
g(x):=\delta D(x)\,u'(x)=C_1-\delta b_0\, H(x-z).
\]
Reflecting boundaries force $u'=0$ at both endpoints, so $g$ vanishes there. Since $z$ is interior, the left endpoint gives $C_1=0$, and the right endpoint gives $C_1-\delta b_0=0$. Hence $\delta b_0=0$, and therefore $g\equiv0$.

Thus $\delta D\,u'=0$ throughout the domain. Away from the source, the steady-state equation gives $(Du')'=\gamma u>0$, so $Du'$ is strictly increasing on each source-free interval and $u'$ has at most one zero there. Wherever $u'\neq0$ this forces $\delta D=0$, and continuity extends the conclusion across those isolated zeros, so $\delta D\equiv0$.
\end{proof}

Source regularity decides whether a continuous kinked alternative is admissible: under known boundaries the It\^o alternatives are excluded by $C^1$ regularity at the source, while the Fickian reflecting case is stronger because the boundary has already removed the flux constant.
\section{Numerical experiments}
\label{sec:practical}

\rev{We now turn from structural to practical identifiability. We empirically test how accurately three physics-informed procedures recover $D(x)$ and $b_0$ from a single finite point pattern. The procedures} differ in how they represent the solution and enforce the PDE and its boundary conditions. Discretize-then-optimize (DTO) differentiates through a finite-volume steady-state solver \cite{karnakov2023SolvingInverseProblems}, so the PDE and boundary conditions hold at the discrete level and accuracy is tied to the mesh; numerical solvers of this kind can rival or outperform neural ones \cite{chatain2025NumericalPDESolvers}. A physics-informed neural network (PINN) \cite{raissiPhysicsinformedNeuralNetworks2019} represents the density and diffusivity as networks and imposes the PDE, jump, and boundary conditions as soft residual penalties, which is mesh-free but makes the fit depend on the penalty weights \rev{(\cref{fig:supp-pinn-wres})}. Bilevel local operator learning (BiLO) \cite{zhangBiLOBilevelLocal2026,zhangBayesianBiLOBilevel2026} learns a local solution operator for the current diffusivity, reducing its reliance on soft penalty weights without a fixed mesh. All three share a variable-projection outer loop that eliminates the source strength analytically. We adapt each to the singular source; \cref{app:numerics-summary} summarizes the common structure, and the full implementation details are in \cref{sec:numerics}.

Structural identifiability is necessary but not sufficient for recovery from finite Poisson data, so we test two identifiable cases: It\^o diffusion with Dirichlet boundaries, where identifiability depends on source regularity (\cref{thm:ito-point-source}), and Fickian diffusion with reflecting boundaries, where the flux ambiguity is removed directly (\cref{cor:neumann-fickian}).

\rev{Fixation gives a single snapshot of each cell, and the diffusivity profile $D(x)$ can vary across cells.} We draw one snapshot per trial and report the median and 10th-to-90th-percentile recovery over 100 independent trials to gauge reliability rather than a single fit. \rev{The snapshots are synthetic draws from the model and each trial has a new random initialization, so the bands show the spread across trials, not a confidence interval for one fit.} Modern training refinements, such as tailored weight initialization \cite{tarbiyati2024WeightInitializationAlgorithm}, adaptive collocation \cite{celaya2025AdaptiveCollocationPoint}, and other training heuristics \cite{wang2023ExpertsGuideTraining}, would likely improve the neural schemes, but we leave that optimization to future work. \rev{We chose the hyperparameters of each method by hand. \Cref{sec:numerics-study} reports their sensitivity, the training curves, and typical failures. We omit the few degenerate fits (less than 1\%) with an unusually large $b_0$, although including them has little effect on the medians (\cref{sec:supp-failure}).}

\begin{figure}[tb]
    \centering
    \includegraphics[width=0.76\linewidth]{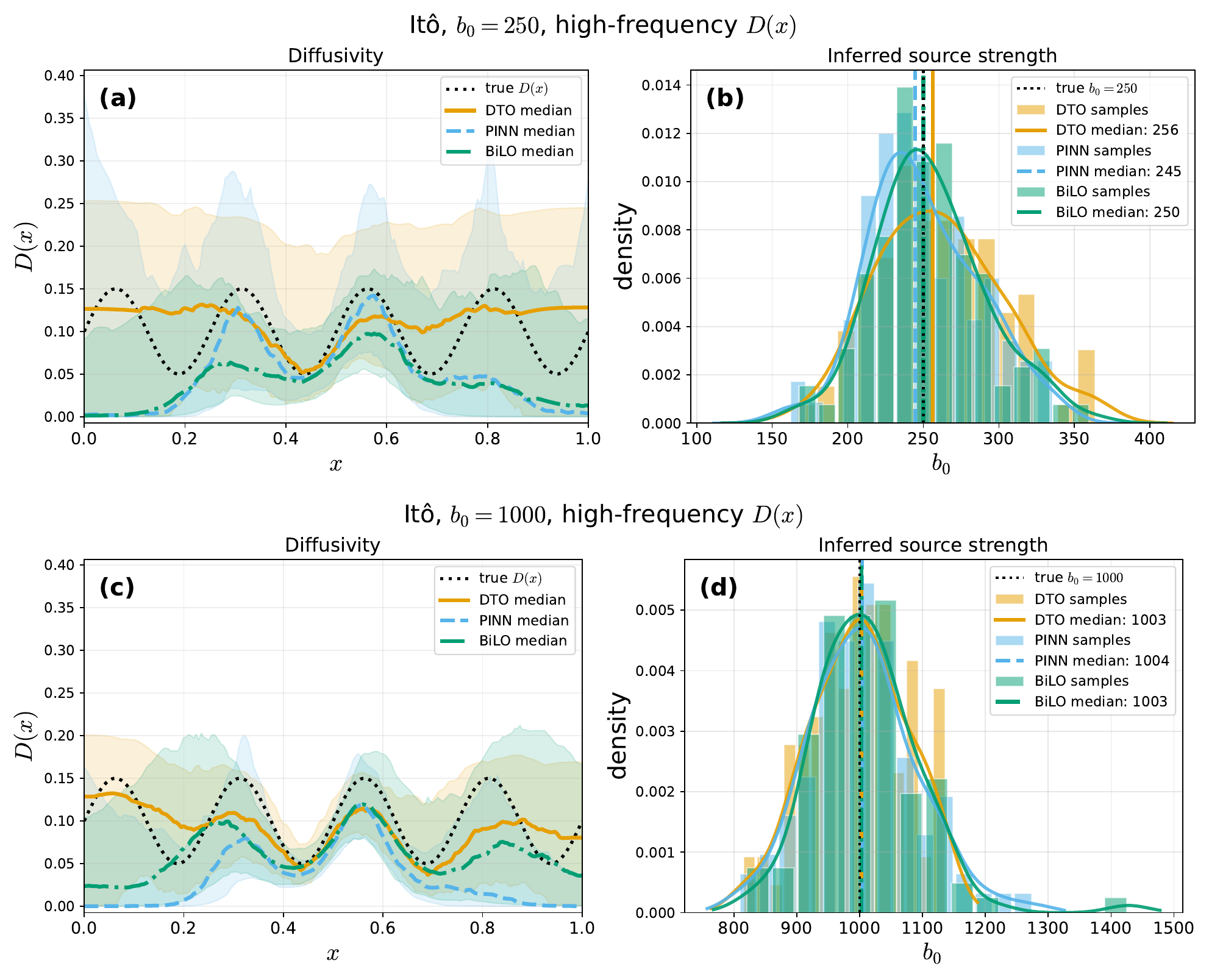}
    \caption{Recovery under It\^o dynamics with Dirichlet boundaries, for $D_{\mathrm{true}}(x) = 0.10 + 0.05\sin(8\pi x)$ at source strengths $b_0 = 250$ (\textbf{a},~\textbf{b}) and $b_0 = 1000$ (\textbf{c},~\textbf{d}). In panels (\textbf{a},~\textbf{c}) the dotted black curve is the true $D(x)$, the colored curves are the per-method median recovered $D(x)$, and the shaded bands span the 10th to 90th percentile over 100 trials. In panels (\textbf{b},~\textbf{d}) the bars are the distribution of the inferred source strength $b_0$ across trials, each normalized to unit area, with kernel density estimates drawn as smooth curves and vertical lines at the medians. \rev{Degenerate fits, with $b_0$ more than 1.5 times its median over trials, are omitted in this and the next two figures.}
}
    \label{fig:Combined_Ito}
\end{figure}

\subsection{It\^o with Dirichlet boundaries}

Across 100 trials per method (\cref{fig:Combined_Ito}), recovery sharpens with the source strength. Because no numerical method imposes the source regularity exactly, this experiment tests whether structural identifiability persists in practice. At $b_0=250$, the source strength and degradation rate give about $48$ particles per snapshot on average, comparable to single-gene transcript counts measured by smFISH \cite{dingConstitutiveSplicingEconomies2019}. Although $b_0$ is recovered reliably, the diffusivity estimate captures only coarse shape features near the source and flattens away from it, with wide 10th to 90th percentile bands that do not resolve the true $D(x)$. At $b_0=1000$, more particles (about $190$ per snapshot) improve all three methods' recovery of $D(x)$ and $b_0$: DTO is the most accurate \rev{in relative $L^2$ error (\cref{tab:supp-errors})}.

\begin{figure}[tb]
    \centering
    \includegraphics[width=0.76\linewidth]{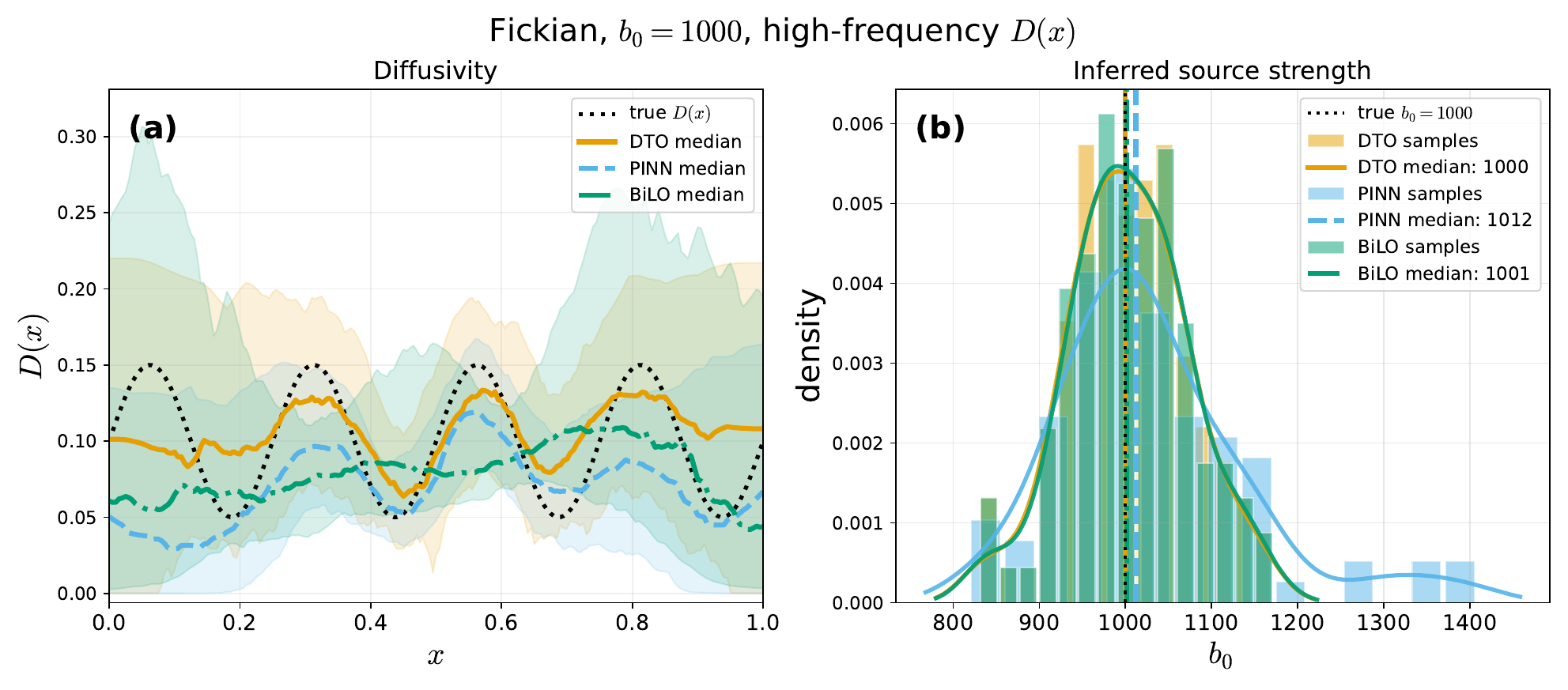}
    \caption{Recovery under Fickian dynamics with Neumann boundaries, for the high-frequency target $D_{\mathrm{true}}(x) = 0.10 + 0.05\sin(8\pi x)$ at $b_0 = 1000$. Even at this source strength only DTO tracks the oscillations of $D(x)$, while PINN and BiLO smooth them out. Plotting conventions follow \cref{fig:Combined_Ito}: (\textbf{a})~median and 10th to 90th percentile band of the recovered $D(x)$ over 100 trials, and (\textbf{b})~inferred birth rate $b_0$ with kernel density estimates.
}
    \label{fig:Combined_Fickian_highfreq}
\end{figure}

\begin{figure}[tb]
    \centering
    \includegraphics[width=0.76\linewidth]{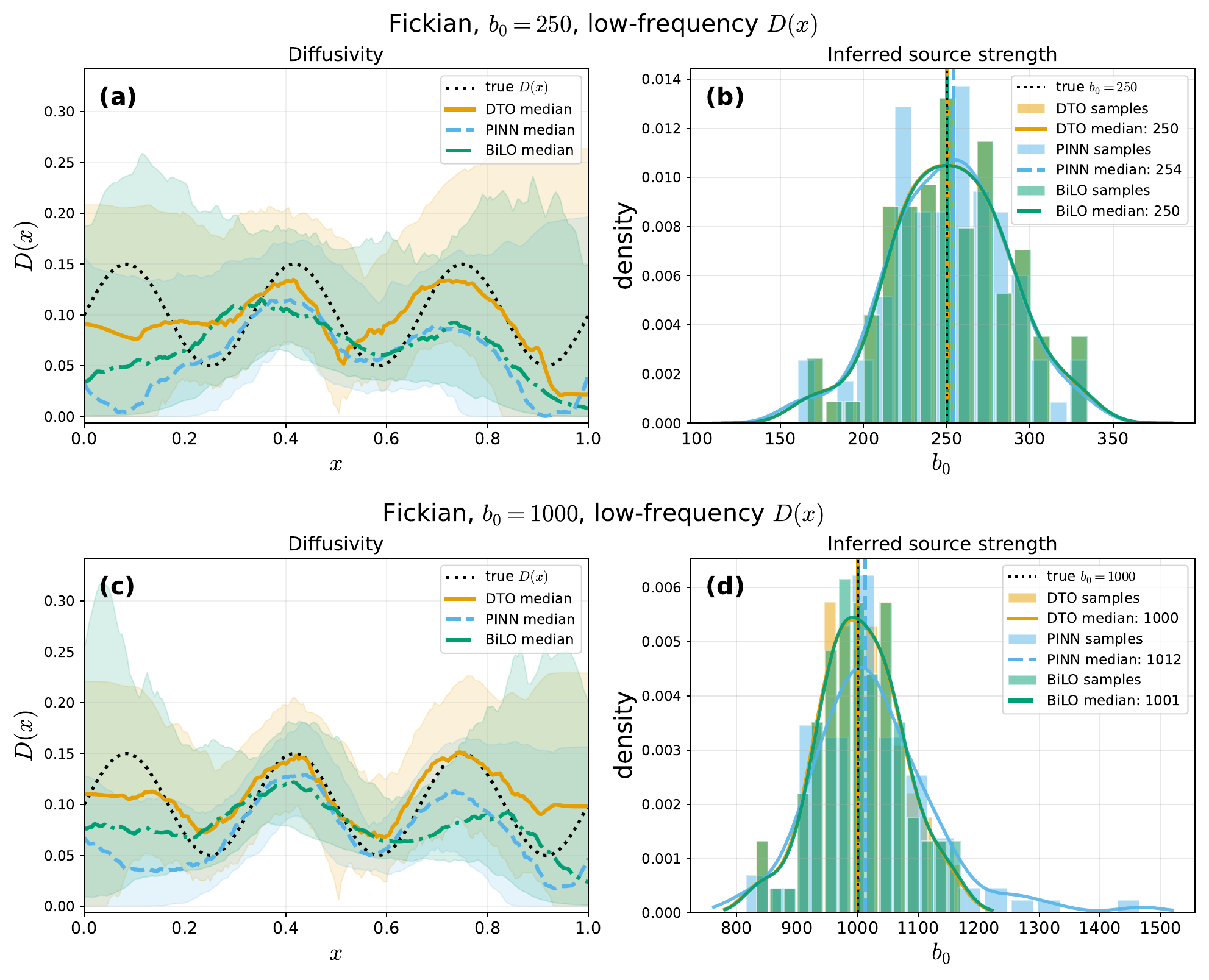}
    \caption{Recovery under Fickian dynamics with Neumann boundaries for the lower-frequency target $D_{\mathrm{true}}(x) = 0.10 + 0.05\sin(6\pi x)$ at $b_0 = 250$ (\textbf{a},~\textbf{b}) and $b_0 = 1000$ (\textbf{c},~\textbf{d}). All three methods recover the oscillations, and the reconstruction tightens as $b_0$ increases. Plotting conventions follow \cref{fig:Combined_Ito}.
}
    \label{fig:Combined_Fickian_lowfreq}
\end{figure}

\subsection{Fickian with reflecting boundaries}

In the Fickian--Neumann case (\cref{fig:Combined_Fickian_highfreq,fig:Combined_Fickian_lowfreq}) the kinked mode is absent, yet recovery of the same high-frequency target $D_{\mathrm{true}}(x)=0.10+0.05\sin(8\pi x)$ \rev{is still difficult. Even} at $b_0=1000$ (about $200$ particles per snapshot) only DTO tracks the oscillations, while PINN and BiLO smooth them out (\cref{fig:Combined_Fickian_highfreq}). Lowering the target to $\sin(6\pi x)$ makes the problem easier, and all three methods then recover it at both $b_0=250$ and $b_0=1000$, about $50$ and $200$ particles per snapshot (\cref{fig:Combined_Fickian_lowfreq}). Higher spatial frequency in $D(x)$ makes the heterogeneity harder to recover.

Across these tests, recovery of both $b_0$ and the spatial shape of $D(x)$ improves with more observed particles, but depends on the method and regularization. \rev{Individual fits are noisy, although at $b_0=1000$ (hundreds of observed particles) they begin to show the qualitative features of $D(x)$ (\cref{fig:supp-snapshots}), with most estimates positively correlated with the true $D(x)$ (\cref{fig:supp-shape-scatter}, \cref{tab:supp-errors}). The median over trials generally follows the spatial variation in the true $D(x)$ more closely than individual fits do.} Two errors recur. First, some smoothing is needed to select among observationally equivalent fields, but too much erases the spatial variation of $D(x)$ \rev{(\cref{fig:supp-smoothness})}; the appropriate regularization is case- and method-dependent. Second, recovered profiles often have a global scale error, $\hat D(x)\approx cD^\ast(x)$, so the shape of $D(x)$ is recovered more reliably than its absolute scale. This reflects a near scaling symmetry of the steady state. With $\gamma = 0$ the balance is homogeneous in $(D, b_0)$, so $(D, b_0)$ and $(cD, cb_0)$ produce the same $u$ and the scale is unidentifiable. Degradation breaks this symmetry, leaving the absolute scale the most weakly constrained aspect, especially in spatial regions where degradation is not the dominant mechanism of particle removal (e.g., boundary exit). For biological applications that aim to locate heterogeneous regions (chromatin domains, nuclear condensates, or interchromatin channels), this shape is the relevant quantity, and absolute diffusivity can be calibrated from complementary measurements.
\section{Unknown boundaries and additional constraints}
\label{sec:extended-structure}

The results so far assumed a known idealized boundary, either absorbing or reflecting. In the biologically standard setup the boundary is partially permeable, and the pore permeability is uncertain. We model this as a Robin boundary with unknown permeability, allowing in general the two endpoint values $\kappa_0,\kappa_L$ to differ. We first show that this extra unknown is by itself enough to make the parameters non-identifiable.

\subsection{Unknown Robin permeability}

Even for It\^o diffusion, an \emph{unknown} boundary permeability re-introduces a degeneracy. Take first a single symmetric permeability $\kappa_0=\kappa_L=\kappa$, the natural extension of the known-boundary cases above. Writing $f=\delta D\,u$, the inherited endpoint conditions $\partial_n f+\delta\kappa\,u=0$ at $x=0,L$ fix the two slopes of $f$ through the single unknown $\delta\kappa$, while the source corner $[f'](z)=-\delta b_0$ links them. Eliminating $\delta\kappa$ gives $\delta b_0=\delta\kappa\,\bigl(u(0)+u(L)\bigr)$.
A nonzero $\delta\kappa$ therefore yields an admissible perturbation $\delta D=f/u$ (still $C^1$ at $z$, since the corner of $f$ cancels the jump in $u'$) that reproduces the same $u$, but only by changing the source strength $b_0$, exactly as in the Dirichlet point-source case. Thus a single unknown permeability leaves a one-parameter ambiguity between $D$ and $b_0$. If the endpoint permeabilities are allowed to differ, the ambiguity persists even when the source strength is fixed.

\begin{theorem}[Robin boundary non-identifiability]
\label{thm:robin-nonident}
For the one-dimensional It\^o model with a single point source and unknown endpoint Robin permeabilities $\kappa_0,\kappa_L$ in the flux condition \eqref{eq:robin-bc}, the parameter set $\{D(x), b_0, \kappa_0,\kappa_L\}$ is non-identifiable from observation of $u(x)$ alone. The degeneracy persists within the $C^1$ admissible class of \cref{ass:regularity}, even with the source strength $b_0$ held fixed: there is a nontrivial one-parameter perturbation of $D$ that is absorbed entirely by changes in the endpoint permeabilities.
\end{theorem}

\begin{proof}
Consider the one-dimensional domain $[0, L]$ with a point source at $z \in (0, L)$:
\begin{equation}
\LL_0 u - \gamma u + b_0 \delta(x - z) = 0.
\end{equation}
Integrating over the domain gives the global mass balance
\begin{equation}
b_0 = \gamma \int_0^L u(x) \, \diff x + \kappa_0 u(0) + \kappa_L u(L),
\end{equation}
in which the permeabilities $\kappa_0,\kappa_L$ appear only through the endpoint terms.

Set $\delta b_0 = 0$ and write $f = \delta D\, u$. By \cref{lem:source-jump}, $f'' = 0$ with no corner at $z$. Choose the homogeneous mode
\[
f(x) = c\,(x - z), \qquad \delta D(x) = \frac{c\,(x - z)}{u(x)}.
\]
Since $f(z) = 0$, the perturbation does not change the source jump.

It remains to check admissibility at the source. Although $u'$ jumps at $z$, the factor $(x-z)$ cancels the induced corner:
\[
\delta D'(x) = \frac{c}{u(x)} - \frac{c\,(x-z)\,u'(x)}{u(x)^2},
\]
so $\delta D'(z^-) = \delta D'(z^+) = c/u(z)$, and \cref{ass:regularity} does not exclude the perturbation.

Finally, subtract the two flux conditions \eqref{eq:robin-bc}. This gives $\partial_n f + \delta\kappa\, u = 0$ at each endpoint, so the permeability perturbations
\[
\delta\kappa_0 = \frac{c}{u(0)}, \qquad \delta\kappa_L = -\frac{c}{u(L)}
\]
make the perturbed boundary conditions hold, and the mass balance is unchanged because $\delta\kappa_0\, u(0) + \delta\kappa_L\, u(L) = c - c = 0$. Thus every sufficiently small $c$ gives a distinct admissible parameter set with the same profile $u$.
\end{proof}
\Cref{fig:Ito_unknownRobin_pointSource} shows the degeneracy for the It\^o model with a point source at $z = 0.5$ and unknown endpoint permeabilities. Starting from $D_1(x) = 2 + 0.5\sin(2\pi x)$, $b_0 = 10$, and $\kappa_0 = \kappa_L = 2$, a different diffusivity $D_2$ and shifted permeabilities reproduce the same density $u_2 \equiv u_1$ \emph{with the source strength unchanged}. Because $\delta D$ is smooth at $z$, the $C^1$ regularity that ensures identifiability under known boundaries (\cref{thm:ito-point-source}) does not help here.

\begin{figure}[tb]
    \centering
    \includegraphics[width=\linewidth]{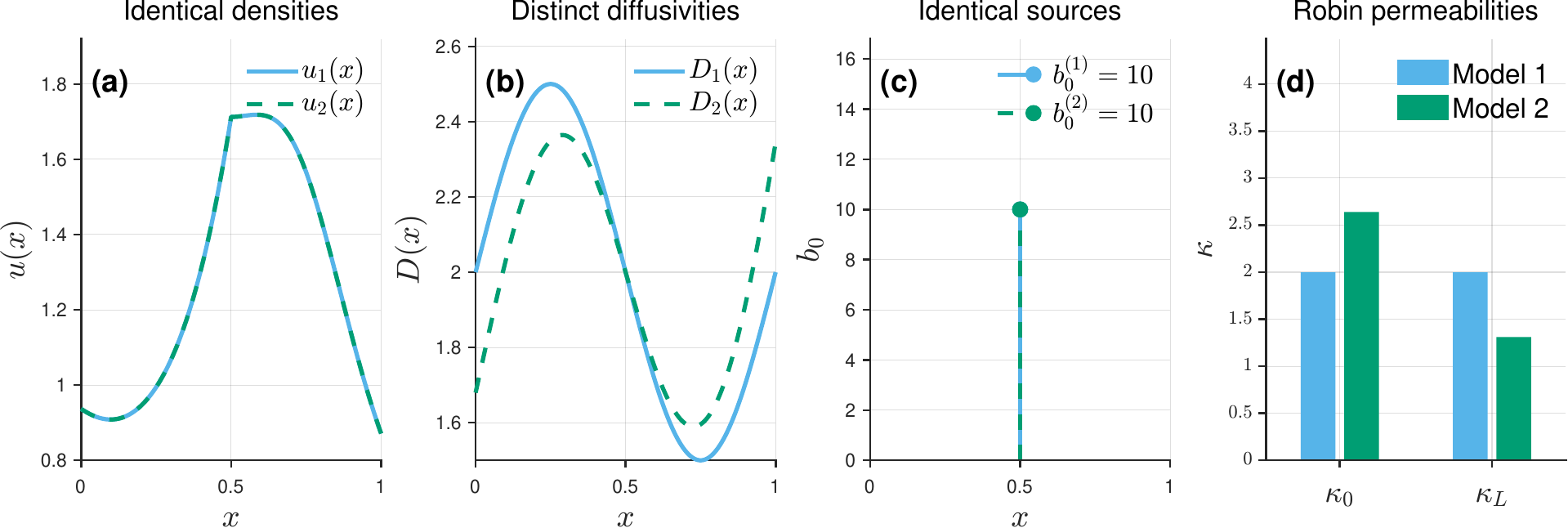}
    \caption{Same-density alternative for It\^o diffusion with unknown Robin permeabilities and fixed source strength. Here $D_1(x) = 2 + 0.5\sin(2\pi x)$, $b_0 = 10$, $z = 0.5$, $\gamma = 5$, and $\kappa_0 = \kappa_L = 2$; the perturbation uses flux offset $c = 0.6$ and follows the construction in \cref{thm:robin-nonident}. Panels show (a)~the common density, (b)~\rev{the diffusivities $D_1$ and $D_2$, which give the same density}, (c)~the unchanged point source, and (d)~the alternative endpoint permeabilities that absorb the flux offset.
}
    \label{fig:Ito_unknownRobin_pointSource}
\end{figure}

Additional structure can nonetheless remove this ambiguity. We consider multiple point sources with shared intensity (\cref{sec:multi-source}), a downstream species supplying a distributed constraint (\cref{sec:downstream-species}), and repeated steady states (\cref{sec:multiple-steady-states}). The statements and mechanisms are given here; the extended algebra is in \cref{app:proof-thm5,app:proof-thm6,app:proof-dual-stimulus}.

\subsection{Multi-source shared-strength constraints}
\label{sec:multi-source}

The simplest setting with unknown boundaries is to observe the system under multiple point sources that share a common intensity.

\begin{proposition}[Multi-source shared-strength constraints]
\label{thm:multi-source}
Consider the one-dimensional It\^o model with a single unknown Robin permeability $\kappa$, unknown diffusivity $D(x)$, and \rev{$K \geq 2$} point sources at distinct known locations \rev{$0 < z_1 < \cdots < z_K < L$}, all with the same unknown strength $b_0$. Under \cref{ass:regularity}, the parameter set $(D, b_0, \kappa)$ is identifiable from the steady-state profile $u$, except in exceptional source configurations.
\end{proposition}

More precisely, write $q_i=-u(z_i)/[u'](z_i)$ and \rev{$m=Ku(0)/(u(0)+u(L))$}. The exceptional configurations are exactly those satisfying all \rev{$K-1$} relations
\[
q_i-q_1=m(z_i-z_1)-\sum_{j<i}(z_i-z_j),\qquad \rev{i=2,\ldots,K}.
\]
A single $\delta b_0$ must satisfy every source jump, so failure of any one of these relations forces the ambiguity to vanish. Biologically, the \rev{shared-strength} setting can arise in G2 cells\rev{, which have replicated their DNA but not yet divided, so} sister chromatids carry two \rev{identical,} spatially separated gene copies with approximately equal transcription rates.

\subsection{Downstream product constraints}
\label{sec:downstream-species}

When only a single point source is available, identifiability with unknown boundaries can be recovered by observing a second, \emph{downstream} species whose source term is determined by the first. The precursor $u(x)$ is produced at a point source and converts to a product $v(x)$, for example nascent pre-mRNA maturing into spliced mRNA, at rate $k$; both species diffuse and degrade. The governing equations are
\begin{align}
\LL_0^{(u)} u - (\gamma_u + k)\, u + b_0\,\delta(x - z) &= 0, \label{eq:precursor} \\
\LL_0^{(v)} v - \gamma_v\, v + k\, u(x) &= 0, \label{eq:product}
\end{align}
where $\LL_0^{(u)}$ and $\LL_0^{(v)}$ denote It\^o diffusion operators with coefficients $D_u(x)$ and $D_v(x)$, respectively.

\begin{proposition}[Downstream product constraints]
\label{thm:precursor-product}
Suppose both steady-state profiles $u(x)$ and $v(x)$ are observed, the kinetic rates $k$, $\gamma_u$, and $\gamma_v$ are known, and the diffusion coefficients satisfy a known structural relationship $D_v(x)=\rho D_u(x)$ for a known constant $\rho>0$. Assume there is no nonzero affine function $\ell$ for which $(\ell u/v)''=0$ on both source-free intervals while $\ell u/v$ also satisfies the precursor source jump and inherited Robin conditions. Then $D_u(x)$, $b_0$, and $\kappa$ are identifiable even with a single point source and unknown Robin permeability.
\end{proposition}

The observed product source $ku$ cannot vary independently, so the same diffusivity perturbation would have to be both source-kinked through $u$ and globally affine through $v$. Dual-color measurements of precursor and mature RNA provide the biological model. The result assumes externally estimated kinetic rates and a known relationship $D_v=\rho D_u$; future work could ask whether the rates and both diffusion fields can be inferred jointly.

\subsection{Multiple steady-state observations}
\label{sec:multiple-steady-states}

Beyond the single-snapshot scenarios above, identifiability can be restored if two distinct steady states are observed from the same underlying diffusivity. We state this result for Fickian diffusion, where the repeated profiles directly overdetermine the flux perturbation.

\begin{proposition}[Multiple steady-state observations]
\label{thm:dual-stimulus}
Consider the one-dimensional Fickian model ($\alpha = 1$) with unknown diffusivity $D(x)$, observed at two distinct steady-state profiles $u_1(x)$ and $u_2(x)$ \rev{produced under the same boundary condition by single point sources at distinct known locations $z_1<z_2$ with unknown strengths. Then} $D(x)$ and the source strengths are identifiable from $u_1$ and $u_2$. If the boundary permeability is an unknown Robin parameter, it is identifiable as well.
\end{proposition}

For each profile, the flux perturbation $g_j=\delta D\,u_j'$ is piecewise constant away from sources. \rev{Between the sources, $u_1$ decreases and $u_2$ increases. The steady-state equation then rules out a nonzero perturbation that keeps both flux perturbations constant there.} Thus both flux perturbations vanish, which fixes $D$, the source strengths, and the Robin permeability.

\rev{Live-cell RNA imaging \cite{xiaSinglemoleculeLivecellRNA2025} could measure both steady-state profiles in the same cell, with transcription active at one locus at a time.}

\section{Discussion}
\label{sec:discussion}

Fitting the observed steady-state density is not enough to identify the underlying coefficients. The regimes, controlled by source topology, stochastic convention, boundary information, and source regularity, are collected in \cref{tab:summary}. Smooth diffuse sources are non-identifiable under every convention, since source variation can be traded against diffusivity variation. For a single point source with known boundary data, the stochastic convention fixes which object the snapshot constrains, through the different orders of the It\^o-type ($\alpha<1$, proved for $\alpha=0$; \cref{rem:intermediate-conventions}) and Fickian perturbation equations, while the boundary condition and source regularity decide whether its remaining constants are fixed. With unknown Robin permeability, even a single point source is non-identifiable, but shared-intensity sources, a precursor-product pair, or repeated Fickian steady states remove the ambiguity under the stated conditions.

\begin{table}[tb]
\centering
\scriptsize
\setlength{\tabcolsep}{3pt}
\renewcommand{\arraystretch}{1.3}
\caption{Summary of identifiability results. Rows with two outcomes distinguish continuous $D$ from $D$ differentiable at the point source.}
\label{tab:summary}
\begin{tabular}{@{}
>{\raggedright\arraybackslash}m{0.098\linewidth}
>{\raggedright\arraybackslash}m{0.095\linewidth}
>{\raggedright\arraybackslash}m{0.098\linewidth}
>{\raggedright\arraybackslash}m{0.530\linewidth}
>{\raggedright\arraybackslash}m{0.095\linewidth}@{}}
\toprule
\textbf{Convention} & \textbf{Source} & \textbf{Boundary} & \textbf{Identifiability} {\normalfont\itshape(\xmark~unidentifiable, \cmark~identifiable)} & \textbf{Result} \\
\midrule
{\color{tableGray}Any $\alpha$}
& diffuse source
& {\color{tableGray}any}
& \xmark\ source--diffusion confounding
& Thm.~\ref{thm:diffuse-source} \\
\midrule
{\color{okabeBlue}Fickian\newline($\alpha=1$)}
& point source
& {\color{okabeOrange}Dirichlet (absorbing)}
& \xmark\ continuous $D$: kinked flux alternative\newline
\cmark\ differentiable $D$ at the source
& Thm.~\ref{thm:fickian-nonident} \\
\addlinespace[0.25em]
{\color{okabeBlue}Fickian\newline($\alpha=1$)}
& point source
& {\color{okabeGreen}Neumann (reflecting)}
& \cmark\ no source regularity needed
& Cor.~\ref{cor:neumann-fickian} \\
\midrule
{\color{okabePurple}It\^o-type\newline($\alpha<1$)}
& point source
& {\color{okabeOrange}Dirichlet (absorbing)}
& \xmark\ continuous $D$: kinked source alternative\newline
\cmark\ differentiable $D$ at the source
& Rem.~\ref{rem:near-null}; Thm.~\ref{thm:ito-point-source} \\
\addlinespace[0.25em]
{\color{okabePurple}It\^o-type\newline($\alpha<1$)}
& point source
& {\color{okabeGreen}Neumann (reflecting)}
& \xmark\ continuous $D$: kinked alternative; $b_0$ fixed by mass balance\newline
\cmark\ differentiable $D$ at the source
& Rem.~\ref{rem:near-null}; Thm.~\ref{thm:ito-point-source} \\
\addlinespace[0.25em]
{\color{okabePurple}It\^o-type\newline($\alpha<1$)}
& point source
& {\color{okabeSky}unknown Robin}
& \xmark\ even with source regularity; permeability absorbs the flux offset
& Thm.~\ref{thm:robin-nonident} \\
\midrule
{\color{okabePurple}It\^o-type\newline($\alpha<1$)}
& \rev{$K\geq2$} sources, shared $b_0$
& {\color{okabeSky}unknown Robin}
& \cmark\ shared strength ties the jump conditions together
& Prop.~\ref{thm:multi-source} \\
\addlinespace[0.25em]
{\color{okabePurple}It\^o-type\newline($\alpha<1$)}
& downstream species
& {\color{okabeSky}unknown Robin}
& \cmark\ product source is observed through the precursor
& Prop.~\ref{thm:precursor-product} \\
\addlinespace[0.25em]
{\color{okabeBlue}Fickian\newline($\alpha=1$)}
& two steady states, same $D$
& {\color{tableGray}any}
& \cmark\ independent profiles overdetermine the flux ambiguity
& Prop.~\ref{thm:dual-stimulus} \\
\bottomrule
\end{tabular}
\end{table}

These results place spatial single-molecule inference within a broader class of snapshot-inference problems. Related examples include fundamental limits in expression-state dynamics \cite{weinreb2018FundamentalLimitsDynamic} and movement--growth confounding in cross-sectional cell-state data \cite{zhang2025InferringStochasticDynamics}. In the spatial setting studied here, the confounding is between diffusivity, source structure, and boundary exchange, and distinct mechanisms can generate the same observed snapshot.

Inference is therefore possible but delicate where the model is identifiable, and the numerical experiments suggest that DTO, PINN, and BiLO can recover the spatial profile of $D(x)$ from a single point pattern. This profile is recovered more reliably than the absolute scale. Such profiles may still be valuable for mapping subcellular heterogeneity, for example around chromatin domains, nuclear condensates, or nuclear speckles that shape local RNA processing \cite{wuDynamicsRNALocalization2024}. This connects to spatial point-process methods that infer spatial structure directly from molecular coordinates \cite{vihrs2022UsingNeuralNetworks, okabe2025MappingIntracellularDynamics}, increasingly at subcellular resolution \cite{wang2025ELLAModelingSubcellular}. The role of the mechanistic model is to say when such spatial structure can be interpreted as a transport coefficient, a source field, a boundary effect, or only an equivalence class of these. Classical identifiability studies in mathematical biology typically concern finite parameter vectors \cite{simpson2020PracticalParameterIdentifiability, browning2020IdentifiabilityAnalysisStochastic, ciocanel2024ParameterIdentifiabilityPDE}\rev{, but recent work asks when unknown functions embedded in mechanistic ODE models are identifiable \cite{loman2026StructuralFunctionalIdentifiability}, a question more akin to ours.}

Our proofs are one-dimensional, where steady states and alternative parameter sets can be written explicitly. \rev{One-dimensional models might already be useful in elongated geometries such as neuronal dendrites, where they have been fitted to steady-state mRNA and protein profiles \cite{fonkeuHowMRNALocalization2019}, although transport there is often active rather than purely diffusive \cite{bressloffStochasticModelsIntracellular2013}.} In two and three dimensions, point sources produce singular solutions, source regularity requires a more careful formulation, and anisotropic diffusion becomes relevant. We expect the ambiguity-operator perspective to carry over, but the details are open. \rev{Applying the methods to real imaging data would need this two- or three-dimensional setting, which adds implementation complexity (for example, adapting the jump condition at the source in the neural methods) beyond the scope of this work.} We have also assumed first-order degradation; nonlinear reactions such as binding, dimerization, or saturation may introduce new identifiable or non-identifiable combinations.

On the statistical side, recovery of shape depends on carefully chosen regularization, and a theory of finite-sample practical identifiability remains to be developed. \rev{DTO gave the most accurate median profiles in our one-dimensional tests and has the fewest hyperparameters, so it is a natural first choice for these one-dimensional problems, although the neural methods may do better with further hyperparameter tuning or in higher dimensions. Bayesian versions of these methods could quantify uncertainty for a single snapshot \cite{zhangBayesianBiLOBilevel2026} and, by modeling how $D(x)$ varies across cells, could pool information from multiple snapshots. Because our synthetic snapshots are drawn from the model, the numerical tests do not address model misspecification or observation error.}

Finally, the stochastic convention should be specified from microscopic or macroscopic physics, since different conventions can fit the same steady-state point pattern. The convention affects escape, search, survival, and exit statistics, including boundary-crossing times relevant to RNA export \cite{tung2025EscapeHeterogeneousDiffusion, tung2026StochasticSearchSpacedependent, liu2026IdentifyingInterpretationTwodimensional}, and distinguishing it requires temporal signatures such as mean-square displacements \cite{liu2026IdentifyingInterpretationTwodimensional}. Future work using perturbation time courses or live-cell imaging could add these constraints and break the steady-state degeneracies identified here \cite{miles2025MechanisticInferenceStochastic}.

\section*{Acknowledgments}
CEM was supported by NSF/NIGMS (NIH) grant DMS-2451263 and NSF CAREER grant DMS-2545859. This work utilized the infrastructure for high-performance and high-throughput computing, research data storage and analysis, and scientific software tool integration built, operated, and updated by the Research Cyberinfrastructure Center (RCIC) at the University of California, Irvine (UCI), \url{https://rcic.uci.edu}. RZZ thanks the NVIDIA Academic Grant Program for the NVIDIA RTX PRO 6000 Blackwell GPU.

\section*{Generative AI declaration}
During the preparation of this work the authors used \rev{generative AI tools, including Claude Opus 4.8/5.5 (Anthropic), to assist with editing the manuscript, refining mathematical proofs, and} implementing portions of the accompanying code. After using \rev{these tools}, the authors reviewed and edited the content as needed. \rev{The authors assume responsibility for all content.}

\appendix
\crefalias{section}{appendix}
\crefalias{subsection}{appendix}

\section{Numerical methods summary}
\label{app:numerics-summary}

The full numerical methods are given in \cref{sec:numerics}\rev{, and this} appendix summarizes their common structure. Given a current estimate of $D$, each method first constructs a response $\hat u$ to a unit source. Since the PDE is linear in the source strength, $b$ can then be determined analytically from the Poisson likelihood \eqref{eq:likelihood}, leaving an update of $D$ under the projected loss. For the unit-source response $\hat u^{(k)}$ associated with the current estimate $D^{(k)}$, the projection and loss are
\begin{align}
b^{(k+1)}&=\frac{N}{\int_\Omega\hat u^{(k)}\,\diff x}, &
\mathcal L_{\mathrm{proj}}[\hat u]
&=N\log\!\left(\int_\Omega\hat u\,\diff x\right)
-\sum_{j=1}^{N}\log\hat u(x_j), \label{eq:compact-varpro}
\end{align}
up to an additive constant. The $D$ update uses
\[
\mathcal J[\hat u,D]=\mathcal L_{\mathrm{proj}}[\hat u]
+\lambda_{\mathrm{smooth}}\mathcal R_{\mathrm{smooth}}(D)
+\lambda_{\mathrm{scale}}\mathcal R_{\mathrm{scale}}(D,\bar D),
\]
where $\mathcal R_{\mathrm{smooth}}$ penalizes roughness and $\mathcal R_{\mathrm{scale}}$ anchors the magnitude of $D$ near the constant-diffusivity pre-fit $\bar D$.

\paragraph{Discretize-then-optimize (DTO)}
It represents $D(x)$ by grid values $D_i$. For $h_i=x_{i+1}-x_i$, set $A_{i+1/2}=2D_i^\alpha D_{i+1}^\alpha/(D_i^\alpha+D_{i+1}^\alpha)$. The conservative discretization is
\begin{equation}
(\mathcal L_{\alpha,D}\hat u)_i\approx\frac{2}{h_{i-1}+h_i}\left[\frac{A_{i+1/2}}{h_i}\left(D_{i+1}^{1-\alpha}\hat u_{i+1}-D_i^{1-\alpha}\hat u_i\right)-\frac{A_{i-1/2}}{h_{i-1}}\left(D_i^{1-\alpha}\hat u_i-D_{i-1}^{1-\alpha}\hat u_{i-1}\right)\right]. \label{eq:compact-fv}
\end{equation}
This conservative discretization gives a tridiagonal unit-source system. DTO obtains $\hat u_D$ from that system and updates $D$ by differentiating $\mathcal J[\hat u_D,D]$ through the solve, so the PDE is enforced by the discrete forward problem.

\paragraph{Physics-informed neural network (PINN)}
$D(x)$ and the unit-source response $\hat u$ are represented by separate neural networks. The source cusp is encoded by $\hat u(x)=\tilde u(x,\phi)$, $\phi=|x-z|$, so $[\hat u'](z)=2\partial_\phi\tilde u(z,0)$. The method minimizes
\begin{equation}
\mathcal L_{\mathrm{PINN}}[\hat u,D]=\mathcal J[\hat u,D]
+w_{\mathrm{res}}\mathcal L_{\mathrm{res}}
+w_{\mathrm{jump}}\mathcal L_{\mathrm{jump}}
+w_{\mathrm{bc}}\mathcal L_{\mathrm{bc}}, \label{eq:compact-pinn-loss}
\end{equation}
jointly over the $\hat u$ and $D$ networks. The three added terms penalize the PDE residual away from $z$, the unit-source jump, and the boundary residuals.

\paragraph{Bilevel local operator learning (BiLO)}
Here $D(x)$ is again a neural network, but the response network is replaced by a neural local solution operator $v[D]$. This separates the physics solve from the data-driven update of $D$ through the lower- and upper-level objectives
\begin{align}
\mathcal L_{\mathrm{op}}[v,D]
&=\mathcal L_{\mathrm{res}}+w_{\mathrm{jump}}\mathcal L_{\mathrm{jump}}
+w_{\mathrm{bc}}\mathcal L_{\mathrm{bc}}+w_{\mathrm{rgrad}}\mathcal L_{\mathrm{rgrad}}, \nonumber\\
v^*[D]&=\arg\min_v\mathcal L_{\mathrm{op}}[v[D],D],
\qquad D^*=\arg\min_D\mathcal J[v^*[D],D]. \label{eq:compact-bilo-loss}
\end{align}
The lower level trains $v[D]$ to satisfy the physics locally in the state $s=(D,D')$, with $\mathcal L_{\mathrm{rgrad}}=\|\nabla_s\boldsymbol{\mathcal R}\|_2^2$ for $\boldsymbol{\mathcal R}=(\mathcal R_{\mathrm{PDE}},\mathcal R_{\mathrm{jump}},\mathcal R_{\mathrm{bc}})$. The upper level updates $D$ using the projected data loss. Code is available at \url{https://github.com/Miles-group-code/AlphaDiffusivityNet}.

\section{Sketch for multi-source shared-strength constraints}
\label{app:proof-thm5}

\begin{proof}[Proof sketch]
Consider \rev{$K \geq 2$} point sources at distinct locations \rev{$z_1, \ldots, z_K \in (0, L)$} with shared intensity $b_0$:
\begin{equation}
\LL_0 u - \gamma u + b_0 \sum_{i=1}^{\rev{K}} \delta(x - z_i) = 0.
\end{equation}
Suppose another admissible triple $(\widetilde D, \widetilde b_0, \widetilde\kappa)$ produces the same profile $u$. Write $f = \delta D\, u$, $\delta b_0 = \widetilde b_0 - b_0$, and $\delta\kappa = \widetilde\kappa - \kappa$. All \rev{$K$} sources share the strength $b_0$, hence the single perturbation $\delta b_0$. By \cref{lem:source-jump}, $f$ is piecewise linear between the $z_i$ with equal corners $[f'](z_i) = -\delta b_0$, and its source values are read off from the observed profile up to that one scale:
\begin{equation}
\label{eq:multi-source-jump}
f(z_i) = \delta b_0\, q_i, \qquad q_i := \frac{D(z_i)\,u(z_i)}{b_0} = -\frac{u(z_i)}{[u'](z_i)}.
\end{equation}

The inherited Robin conditions $\partial_n f + \delta\kappa\, u = 0$ at the two endpoints involve the single permeability perturbation $\delta\kappa$, so they fix the endpoint slopes of $f$ in terms of $\delta b_0$ alone. Eliminating $\delta\kappa$ between them gives \rev{$K\delta b_0$} $= \delta\kappa\,(u(0)+u(L))$, and the leftmost slope is $f'(0^+) = m\,\delta b_0$ with \rev{$m := K u(0)/(u(0)+u(L))$}.

Suppose $\delta b_0 \neq 0$. Then $f/\delta b_0$ is the piecewise-linear function with slope $m$ on $(0,z_1)$, dropping by one at each source, and one free intercept $a$. Evaluating at the sources,
\begin{equation}
\frac{f(z_i)}{\delta b_0} = a + m\, z_i - \sum_{j<i}(z_i - z_j) = q_i, \qquad i = 1,\dots,\rev{K}.
\end{equation}
These \rev{$K$} equations in the one constant $a$ leave \rev{$K-1$} compatibility relations among the observed $q_i$ and the source spacings. The first two sources already require
\begin{equation}
q_2 - q_1 = (m-1)(z_2 - z_1).
\end{equation}
Such relations are not imposed by the PDE or the boundary condition, so they hold only for pathological source locations and profiles.

When they fail, $\delta b_0 = 0$, and \rev{$K\delta b_0$} $= \delta\kappa\,(u(0)+u(L))$ forces $\delta\kappa = 0$. With no source corners and zero endpoint slopes, $f$ is constant, and $f(z_i) = \delta b_0\, q_i = 0$ makes that constant zero. Hence $f \equiv 0$ and $\delta D \equiv 0$.
\end{proof}

\section{Sketch for downstream product constraints}
\label{app:proof-thm6}

\begin{proof}[Proof sketch]
Consider the two-species system with precursor $u$ and product $v$:
\begin{align}
\LL_0^{(u)} u - (\gamma_u + k) u + b_0 \delta(x - z) &= 0, \label{eq:precursor-pde} \\
\LL_0^{(v)} v - \gamma_v v + k u(x) &= 0, \label{eq:product-pde}
\end{align}
where $\LL_0^{(u)} u = (D_u u)''$ and $\LL_0^{(v)} v = (D_v v)''$ are the It\^o diffusion operators for each species. We assume the structural constraint $D_v = \rho D_u$ for some known constant $\rho > 0$.

Suppose $(\tilde{D}_u, \tilde{b}_0, \tilde{\kappa})$ is an alternative parameter set producing the same observed densities $u(x)$ and $v(x)$. Define $\delta D_u = \tilde{D}_u - D_u$, $\delta b_0 = \tilde{b}_0 - b_0$, and $\delta \kappa = \tilde{\kappa} - \kappa$.

For the precursor, \cref{lem:source-jump} gives the auxiliary $f_u = \delta D_u\, u$ with $f_u'' = -\delta b_0\,\delta(x-z)$, piecewise linear with a corner at $z$; with the boundary unknown this leaves a one-parameter family (\cref{thm:robin-nonident}). The product equation supplies the missing constraint, and is the part special to this case. Subtracting \cref{eq:product-pde} for the original and perturbed parameters (noting that the source $ku(x)$ is \emph{observed} and hence identical in both):
\begin{equation}
\label{eq:product-ambiguity}
(\delta D_v \cdot v)'' = 0.
\end{equation}
Under the structural assumption $D_v = \rho D_u$, we have $\delta D_v = \rho \, \delta D_u$, so
\begin{equation}
(\delta D_u \cdot v)'' = 0.
\end{equation}

Define $f_u = \delta D_u\, u$ and $\psi_v = \delta D_u\, v$. The precursor equation makes $f_u$ piecewise linear with one derivative jump at $z$, fixed by $\delta b_0$. The product equation makes $\psi_v$ globally linear, since $\psi_v'' = 0$ with no source, so $\psi_v(x) = ax + b$. Because the same diffusivity perturbation appears in both species, $\delta D_u = f_u/u = \psi_v/v$, hence
\begin{equation}
f_u(x) = \frac{(ax + b)\, u(x)}{v(x)}.
\end{equation}
For $f_u$ to be an admissible precursor perturbation, the right-hand side must be linear on each of $(0,z)$ and $(z,L)$, carry the prescribed derivative jump at $z$, and meet the inherited Robin conditions. A nonzero ambiguity therefore requires a nonzero affine $\ell(x) = ax + b$ with
\begin{equation}
\bigl(\ell(x)\, u(x)/v(x)\bigr)'' = 0 \quad \text{on } (0,z) \text{ and } (z,L),
\end{equation}
together with the source-jump and boundary matching. This asks the observed ratio $u/v$ to turn a global affine function into a kinked piecewise-affine one with exactly the right matching data, which holds only for exceptional pairs $(u,v)$. Otherwise no such nonzero $\ell$ exists, so $\psi_v \equiv 0$, and since $v > 0$ in the interior $\delta D_u \equiv 0$. The precursor jump then gives $\delta b_0 = 0$, and the Robin boundary condition gives $\delta\kappa = 0$.
\end{proof}

\section{Proof of multiple steady-state observations}
\label{app:proof-dual-stimulus}

\begin{proof}
Suppose another admissible parameter set produces the same two profiles $u_1,u_2$, and let $\delta D$ be the common diffusivity perturbation. The source strengths are unknowns and may differ between the two sets. By \cref{lem:flux-gauge}, each profile gives a piecewise-constant flux perturbation
\begin{equation}
g_1 = \delta D\, u_1', \qquad g_2 = \delta D\, u_2',
\end{equation}
\rev{with $g_j$ constant on either side of $z_j$. Since $(Du_j')'=\gamma u_j>0$ away from $z_j$ and the endpoint fluxes satisfy $Du_j'(0)\geq0$ and $Du_j'(L)\leq0$, $u_j$ increases to the left of its source and decreases to the right. On $(z_1,z_2)$, $g_1Du_2'=g_2Du_1'$. Differentiating gives $g_1u_2=g_2u_1$. The profiles are positive and their derivatives have opposite signs, so both constants vanish. Thus $g_2=0$ on $(0,z_2)$ and $g_1=0$ on $(z_1,L)$. The corresponding derivatives are nonzero on these intervals, which together cover $(0,L)$, so $\delta D\equiv0$.}

The jumps of $g_j$ across the sources are the source-strength perturbations, which therefore vanish. Finally, with $\delta D \equiv 0$, subtracting the Robin boundary condition gives $\delta\kappa\, u_j = 0$ at the boundary, and since the steady states are positive there, $\delta\kappa = 0$. The ambiguity is removed in the interior, independently of the boundary, which is why this setting tolerates an unknown permeability.
\end{proof}

\section{Numerical methods}
\label{sec:numerics}

We compare three reconstruction methods that differ in how they represent the steady-state solution and enforce the PDE, boundary condition, and source jump. DTO uses a differentiable finite-volume solve, PINN uses neural-network residual penalties, and BiLO uses a neural local solution operator. All three methods use the same Poisson point-process likelihood and regularization. In each case we eliminate the source strength analytically, so the remaining numerical optimization is over $D(x)$.

\subsection{Variable projection}

Throughout, $u$ denotes the physical intensity. The unit-source responses are $\hat u_D$ for DTO, $\hat u$ for PINN, and $v[D]$ for BiLO; they are compared to the observed point pattern only through $\mathcal{L}_{\mathrm{proj}}$ (defined below), after the source intensity $b$ has been projected out.

\paragraph{Data loss and source projection}
The data loss is the negative log of the Poisson likelihood \eqref{eq:likelihood} for $m_{\mathrm{obs}}$ snapshots containing $N_{\mathrm{obs}}$ points in total, $X = \{x_1, \dots, x_{N_{\mathrm{obs}}}\}$,
\[\mathcal{L}_{\mathrm{data}}[u] = m_{\mathrm{obs}}\int_{\Omega} u(x) \, \diff x - \sum_{j=1}^{N_{\mathrm{obs}}} \log \left( u(x_j) \right).\]
For a candidate $D$, let $\hat{u}_D$ solve the steady-state equation \eqref{eq:steady-state} with operator $\mathcal{L}_\alpha$ from \eqref{eq:alpha-operator}, a unit point source at $z$, and the chosen boundary condition $\mathcal{B}[\hat{u}]=0$; here $\mathcal{B}$ is the boundary operator for that condition (Dirichlet, Neumann, or Robin; see \eqref{eq:alpha-flux} and \eqref{eq:robin-bc}). Since the equation is linear in $u$, the density at source strength $b$ is $u = b\hat{u}_D$. Substituting $u = b\hat{u}$ into the data loss,
\[\mathcal{L}_{\mathrm{data}}[b\hat{u}] = m_{\mathrm{obs}}\, b \int_{\Omega} \hat{u}(x) \, \diff x - \sum_{j=1}^{N_{\mathrm{obs}}} \left( \log b + \log \hat{u}(x_j) \right),\]
and setting $\partial \mathcal{L}_{\mathrm{data}} / \partial b = 0$ gives the optimal intensity in closed form,
\[b^* = \frac{N_{\mathrm{obs}}}{m_{\mathrm{obs}}\int_{\Omega} \hat{u}(x) \, \diff x}.\]
We take $m_{\mathrm{obs}}=1$ in all experiments below.
Substituting $b^*(\hat u)$ gives the projected loss $\mathcal{L}_{\mathrm{proj}}[\hat u] = \mathcal{L}_{\mathrm{data}}[b^*(\hat u)\hat u]$, up to an additive constant independent of $\hat u$,
\[
\mathcal{L}_{\mathrm{proj}}[\hat u]
=
N_{\mathrm{obs}}\log\!\left(\int_\Omega \hat u(x)\,\diff x\right)
-
\sum_{j=1}^{N_{\mathrm{obs}}}\log \hat u(x_j).
\]
Eliminating $b$ this way is the variable-projection step \cite{golub2003SeparableNonlinearLeast}, and leaves $D$ as the remaining unknown. Starting from an initial field $D^{(0)}$ with unit response $\hat{u}^{(0)}$, we alternate
\begin{enumerate}
  \item the analytic projection $b^{(k+1)} = N_{\mathrm{obs}} \big/ \big(m_{\mathrm{obs}}\int_{\Omega} \hat{u}^{(k)}(x) \, \diff x\big)$, and
  \item the constrained update
  \begin{equation}
    \begin{aligned}
      &  D^{(k+1)}(x) = \arg\min_{D}  \mathcal{L}_{\mathrm{proj}}[\hat{u}] \\
      \text{s.t.} &
      \begin{cases}
      \quad \mathcal{L}_{\alpha,D} \hat{u} - \gamma\hat{u} = - \delta(x - z)\\
      \quad \mathcal{B}[\hat{u}] = 0.
      \end{cases}
    \end{aligned}
  \end{equation}
\end{enumerate}
The three methods below differ in how they carry out step 2.

\paragraph{Shared regularization}
For a unit-source response $q$, define the regularized objective
\[
\mathcal{J}[q, D] = \mathcal{L}_{\mathrm{proj}}[q] + \lambda_{\mathrm{smooth}}\,\mathcal{R}_{\mathrm{smooth}}(D) + \lambda_{\mathrm{scale}}\,\mathcal{R}_{\mathrm{scale}}(D, \bar D),
\]
with $q = \hat u_D$ for DTO, $q = \hat u$ for PINN, and $q = v[D]$ for BiLO. The diffusion-field regularizers are
\[\mathcal{R}_{\mathrm{smooth}}(D) = \int_\Omega \big(D'(x)\big)^2 \, \diff x \quad (H^1),\]
\[\mathcal{R}_{\mathrm{scale}}(D, \bar D) = \int_\Omega \big(D(x) - \bar D\big)^2 \, \diff x,\]
evaluated as grid averages on the solver or collocation grid. \rev{The scale regularizer $\mathcal{R}_{\mathrm{scale}}$ penalizes deviations of $D$ from the constant pre-fit value $\bar D$} to avoid degenerate solutions in the optimization.

\paragraph{Scalar pre-fit}
When the diffusion field $D(x)$ is a constant $D$, the $\alpha$-operator collapses to $D\partial_{xx}$ for every convention, and the unit-source response solves
\[
D \hat u''(x) - \gamma \hat u(x) = -\delta(x - z).
\]
The objective then reduces to a function of two scalars, the constant $D$ and the source strength $b$, which we minimize numerically by alternating a projection step on $b$ with a gradient step on $D$. The resulting fit $\bar D$ provides the anchor in $\mathcal{R}_{\mathrm{scale}}(D, \bar D)$ and the initial guess for the full field $D(x)$.

\paragraph{Interface form of the point source}
The unit-source response is continuous at $z$, but its derivative jumps there. PINN and BiLO therefore impose the singular forcing through the equivalent interface conditions \cite{levequeImmersedInterfaceMethod1994}:
\begin{equation} \label{eq:interface}
  \LL_\alpha \hat{u} - \gamma\hat{u} = 0,\quad
  [\hat{u}](z)= 0,\quad
  D(z)[\hat{u}'](z) = -1,\quad
  \mathcal{B}[\hat{u}] = 0.
\end{equation}
Here $[f](z)=f(z^+)-f(z^-)$, where $+$ and $-$ denote limits from the right and left of $z$. DTO and the synthetic-data solver instead use a discrete approximation of the point source on the finite-volume grid.

\subsection{Discretize-then-optimize (DTO)}

The discretize-then-optimize approach reconstructs $D(x)$ by differentiating through a numerical steady-state solver, an instance of optimizing a discrete loss \cite{karnakov2023SolvingInverseProblems}. We discretize $[0,L]$ on a grid with spacings $h_i = x_{i+1}-x_i$ and approximate the flux-form operator with a conservative finite-volume scheme, using the harmonic mean of $D^\alpha$ at cell interfaces,
\[A_{i+1/2} = \frac{2\,D_i^\alpha D_{i+1}^\alpha}{D_i^\alpha + D_{i+1}^\alpha},\]
which preserves flux continuity across interfaces. The resulting tridiagonal system is solved by a differentiable Thomas algorithm, allowing gradients of the projected data loss to backpropagate through the solver.

We parameterize $D$ directly as $D(x) = \theta(x) + D_{\min}$, with $\theta$ clamped to $\theta \ge 0$ after each step so that $D \ge D_{\min} > 0$. This keeps $\partial D/\partial\theta = 1$, avoiding the vanishing Jacobian of $D = \exp(\theta)$ or $\mathrm{softplus}(\theta)$ as $D \to 0$, which stalls optimization in low-diffusivity regions. DTO then minimizes the shared objective $\mathcal{J}[\hat u_D, D]$, where $\hat u_D$ is the unit-source response from the differentiable finite-volume solve, so the physics is enforced through the discrete forward problem rather than through residual penalties.

\subsection{Physics-informed neural networks (PINNs)}

PINNs use neural networks as mesh-free approximations for the unit-source response $\hat u(x)$ and $D(x)$ \cite{raissiPhysicsinformedNeuralNetworks2019,cuomoScientificMachineLearning2022}. Instead of solving the steady-state PDE exactly, a PINN enforces it as a soft penalty: the squared PDE, jump, and boundary residuals are summed with the projected data loss into a single weighted objective, minimized over $\hat u$ and $D$ together. We define these terms next.

\paragraph{Cusp-capturing network}
An ordinary neural network that is smooth in $x$ cannot represent this derivative jump exactly \cite{tsengCuspcapturingPINNElliptic2023}. We therefore write $\hat{u}(x)=\tilde{u}(x, \phi)$, where $\phi = |x - z|$.
This transformation keeps the function continuous at the source while letting its spatial derivative jump. With $+$ denoting the right limit and $-$ the left limit, the one-sided derivatives are
\[
\hat{u}'(z^\pm) = \partial_x \tilde{u}(z, 0) \pm \partial_\phi \tilde{u}(z, 0).
\]
Consequently, the jump in the derivative is
\[[\hat{u}'](z) = 2\partial_\phi \tilde{u}(z, 0).\] 

We define the residual of the jump condition $\mathcal{R}_{\mathrm{jump}}$ 

\[\mathcal{R}_{\mathrm{jump}}[\hat u, D] = 2D(z)\partial_\phi \tilde{u}(z, 0) + 1.\]

\paragraph{Loss functions}

For notational simplicity, we consolidate the governing physics (PDE, jump, and boundary conditions) into a unified residual operator
$\boldsymbol{\mathcal{R}}[\hat u, D]$ as:
\begin{equation}
  \boldsymbol{\mathcal{R}}[\hat u, D] = \begin{bmatrix}
    \mathcal{R}_{\mathrm{PDE}}[\hat u, D] \\
    \mathcal{R}_{\mathrm{jump}}[\hat u, D] \\
    \mathcal{R}_{\mathrm{bc}}[\hat u, D]
  \end{bmatrix}
\end{equation}
where $\mathcal{R}_{\mathrm{PDE}}[\hat u, D](x) = \mathcal{L}_{\alpha, D} \hat u(x) - \gamma \hat u(x)$, for $x \in \Omega \setminus \{z\}$,
$\mathcal{R}_{\mathrm{jump}}[\hat u, D] = 2D(z)\partial_\phi \tilde{u}(z, 0) + 1$, and $\mathcal{R}_{\mathrm{bc}}[\hat u, D] = \mathcal{B}[\hat u](x)$ for $x \in \partial \Omega$.
Physical consistency is then the condition $\boldsymbol{\mathcal{R}}[\hat u, D] = \mathbf{0}$.

The PDE loss averages the squared residual over the interior collocation points $x_i \neq z$ (the source is handled by the jump residual), the jump loss is the squared source residual, and the boundary loss is the mean of the two endpoint residuals,
\[
\begin{aligned}
\mathcal{L}_{\mathrm{res}}[\hat u, D] &= \frac{1}{N_c} \sum_{i=1}^{N_c} \left( \mathcal{R}_{\mathrm{PDE}}[\hat u, D](x_i) \right)^2, \\
\mathcal{L}_{\mathrm{jump}}[\hat u, D] &= \left( \mathcal{R}_{\mathrm{jump}}[\hat u, D] \right)^2, \\
\mathcal{L}_{\mathrm{bc}}[\hat u, D] &= \frac{1}{2} \mathcal{B}[\hat u](0)^2 + \frac{1}{2} \mathcal{B}[\hat u](L)^2 .
\end{aligned}
\]
In both neural-network methods the boundary penalty is used only for Neumann conditions. Dirichlet conditions are imposed exactly by multiplying the network output by $x(L-x)$, so they need no boundary penalty. The full PINN objective adds these to the projected data loss and regularization,
\[
\mathcal{L}_{\mathrm{PINN}}[\hat u, D] = \mathcal{J}[\hat u, D] + w_{\mathrm{res}}\mathcal{L}_{\mathrm{res}}[\hat u, D] + w_{\mathrm{jump}}\mathcal{L}_{\mathrm{jump}}[\hat u, D] + w_{\mathrm{bc}}\mathcal{L}_{\mathrm{bc}}[\hat u, D],
\]
minimized over $\hat u$ and $D$ jointly.

\paragraph{Implementation}

We parameterize the unit-source response as $\hat u_{\theta_u}(x)=\tilde u_{\theta_u}(x,|x-z|)$ and the diffusivity as the exponential of a neural network, $D_{\theta_D}(x) = \exp\!\left(\mathrm{NN}_{\theta_D}(x)\right)$, which keeps $D$ positive automatically. The spatial derivatives entering the residuals ($\hat u'$, $\hat u''$, and the derivatives of $D$) are obtained by automatic differentiation \cite{paszkePyTorchImperativeStyle2019}.

\subsection{Bilevel local operator learning (BiLO)}

Standard PINNs minimize a weighted sum of data loss and PDE loss, so the fit depends on the relative weights. BiLO \cite{zhangBiLOBilevelLocal2026,zhangBayesianBiLOBilevel2026} separates the data and physics objectives by training a local solution operator that satisfies the PDE for the current $D(x)$ and has the correct local sensitivity to changes in $D$. Let the local diffusion state be $s = (D, D') \in \mathbb{R}^2$, and let $v(x; s)$ denote the local operator, so that
\[
\hat{u}(x) = v\big(x;\, s(x)\big), \qquad s(x) = \big(D(x),\, D'(x)\big).
\]
We write this as $v[D](x) = v(x; s(x))$ when emphasizing the dependence on $D$.

The operator is trained so that, at the current field, the physics residual vanishes and is locally stationary in the state variables:
\[
\boldsymbol{\mathcal{R}}[v[D], D] = \mathbf{0},
\qquad
\nabla_{s}\boldsymbol{\mathcal{R}}[v[D], D] = \mathbf{0},
\]
where $\nabla_{s} = (\partial_D, \partial_{D'})$ differentiates the free state and is then evaluated at $s = s(x)$.

\paragraph{Local-operator loss}

The local operator reuses the PINN physics losses $\mathcal{L}_{\mathrm{res}}$, $\mathcal{L}_{\mathrm{jump}}$, $\mathcal{L}_{\mathrm{bc}}$ (now evaluated on $v[D]$) and adds a residual-gradient (sensitivity) term that enforces $\nabla_s\boldsymbol{\mathcal{R}}=\mathbf{0}$:
\[
\begin{aligned}
\mathcal{L}_{\mathrm{rgrad}}[v,D]
&= \frac{1}{N_c}\sum_{i=1}^{N_c}\big\|\nabla_s \mathcal{R}_{\mathrm{PDE}}[v,D](x_i)\big\|^2
+ \big\|\nabla_s \mathcal{R}_{\mathrm{jump}}[v,D]\big\|^2 \\
&\quad + \tfrac{1}{2}\!\!\sum_{x_b\in\{0,L\}}\!\!\big\|\nabla_s \mathcal{R}_{\mathrm{bc}}[v,D](x_b)\big\|^2 .
\end{aligned}
\]
The local-operator loss is then
\[
\mathcal{L}_{\mathrm{op}}[v,D]
= \mathcal{L}_{\mathrm{res}} + w_{\mathrm{jump}}\mathcal{L}_{\mathrm{jump}}
+ w_{\mathrm{bc}}\mathcal{L}_{\mathrm{bc}}
+ w_{\mathrm{rgrad}}\,\mathcal{L}_{\mathrm{rgrad}} .
\]

\paragraph{Bilevel optimization}
BiLO solves the bilevel problem
\[
\begin{aligned}
D^* &= \arg\min_{D} \mathcal{J}[v^*[D], D], \\
v^*[D] &= \arg\min_{v} \mathcal{L}_{\mathrm{op}}[v[D], D],
\end{aligned}
\]
in which the lower level trains the local operator for the current $D$, and the upper level updates $D$ to minimize the projected data loss evaluated on that operator. We parameterize the diffusion field and the operator as neural networks, with weights $\theta_D$ and $\theta_v$. Because $\mathcal{J}$ depends on $\theta_D$ through the local operator $v[D]$, its gradient flows through $v$ by the chain rule. At each iteration we take one gradient step on the operator, and a gradient step on the diffusion field only when the lower-level loss $\mathcal{L}_{\mathrm{op}}$ is below a tolerance $\epsilon$. Each update of $D$ is therefore based on an operator whose lower-level loss is below the prescribed tolerance, approximating the bilevel structure without solving the lower level to convergence.

\paragraph{Implementation}
All derivatives are computed by automatic differentiation in PyTorch. The diffusion network maps $x$ to $D$ and $D'$, with $D$ made positive by the same exponential map as the PINN, and the operator gives $v = v(x, D, D')$; differentiating in $x$ gives $v'$ and $v''$ for the residual $\mathcal{R} = F(v'', v', v, D, D')$. Autograd gives $\partial_D \mathcal{R}$ directly. For $\partial_{D'}\mathcal{R}$, however, the input $D'$ and the chain-rule factor $D'$ in the total derivative $v' = v_x + v_D\, D' + v_{D'}\, D''$ are separate graph nodes. Differentiating only with respect to the input node misses the product-rule terms $\partial_{v'}\mathcal{R}\, v_D$ and $2\,\partial_{v''}\mathcal{R}\, v_D'$. We add these two terms explicitly.

\begin{figure}[tbp]
  \centering
  \includegraphics[width=0.66\linewidth]{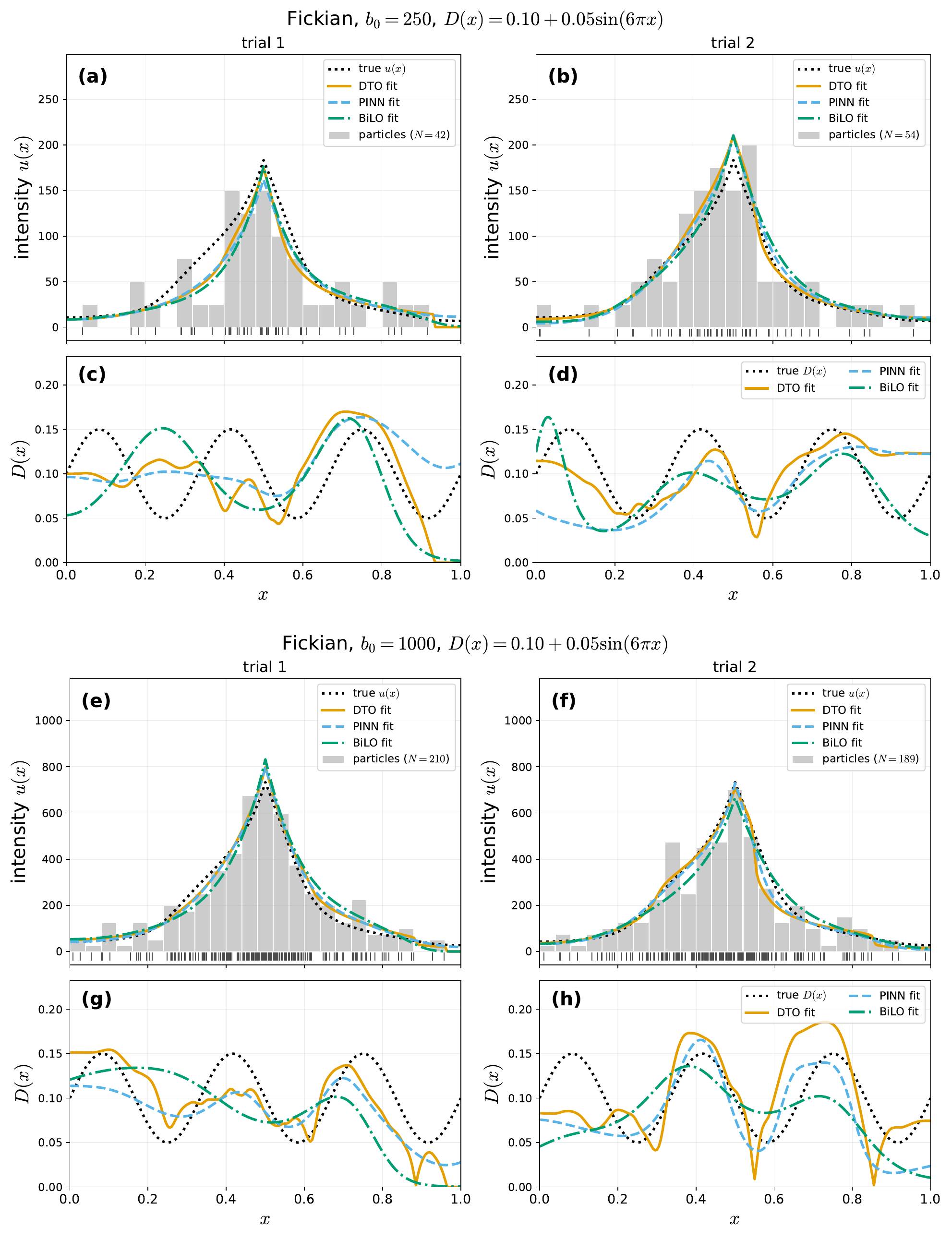}
  \caption{\rev{Example snapshots and fits for Fickian dynamics with Neumann boundaries and $D_{\mathrm{true}}(x)=0.10+0.05\sin(6\pi x)$, two trials each at $b_0=250$ (\textbf{a}--\textbf{d}) and $b_0=1000$ (\textbf{e}--\textbf{h}). Upper panels: particle positions (ticks, and a histogram scaled to intensity), the true intensity, and the intensity re-solved from each method's fitted $D$ and $b_0$. Lower panels: the fitted $D(x)$. The re-solved intensities follow the data closely even where the fitted $D$ differs from the truth.}}
  \label{fig:supp-snapshots}
\end{figure}

\subsection{Synthetic data}
Each experiment fixes a ground-truth pair $(D_{\mathrm{true}}, b_{\mathrm{true}})$ and a source location $z$. We solve the steady-state equation \eqref{eq:steady-state} on a grid with the same conservative finite-volume scheme used by DTO and form the true density $u_{\mathrm{true}} = b_{\mathrm{true}}\hat{u}$. Each snapshot is a draw from an inhomogeneous Poisson point process with intensity $u_{\mathrm{true}}$: the number of molecules is drawn as $N\sim\mathrm{Poisson}\!\big(\int_\Omega u_{\mathrm{true}}\,\diff x\big)$ and their positions by inverse-CDF sampling from the piecewise-linear density on that grid; the experiments reported here use a single snapshot. Because the generative model is exactly the observation model of \cref{sec:framework}, the data loss is its true negative log-likelihood. \rev{\Cref{fig:supp-snapshots} shows example snapshots at $b_0=250$ and $b_0=1000$, with the density re-solved from each method's fitted $D$ and $b_0$, and the fitted $D$ itself.}

\subsection{Experimental setup}
All experiments use the unit interval $[0,1]$, decay rate $\gamma = 5$, a single source at $z = 0.5$, and a solver and collocation grid of $n_{\mathrm{res}} = 201$ points. The neural-network methods (PINN and BiLO) represent $D(x)$ by a width-64, depth-3 MLP and the unit-source response or local operator by a width-256 modified MLP, both using random-Fourier-feature input embeddings \cite{wangEigenvectorBiasFourier2021}, which map the input through $\Phi_k(x)=\sqrt{2/m}\,\sin(\omega_k x+\beta_k)$ with random frequencies $\omega_k\sim\mathcal{N}(0,\sigma^2)$ and phases $\beta_k\sim\mathrm{Unif}[0,2\pi]$ at scale $\sigma=5$. Training uses a 500-step constant-$D$ scalar pre-fit, \rev{which all three methods share as their starting point. Following \cite{zhangBiLOBilevelLocal2026}, PINN and BiLO then pretrain before the data loss is switched on: PINN holds $D$ at the pre-fit value and trains $\hat u$, and BiLO alternates steps that fit its $D$ network to the pre-fit with steps that train the operator. In both, the solution network is trained on the physics residuals and on the finite-volume solution for the pre-fit $D$, for 1000 steps (It\^o) or 2000 steps (Fickian). All methods then use} Adam with a cosine schedule, early stopping, and at most 7000 steps. The physical parameters of each test (stochastic convention $\alpha$, boundary type, $D_{\mathrm{true}}(x)$, and source strength $b_{\mathrm{true}}$) are stated in the corresponding figure captions; the per-method loss weights and learning rates are given in \rev{\cref{tab:supp-hyperparams}}.

\section{Numerical study}
\label{sec:numerics-study}

\rev{The runs in this section use the same synthetic snapshots (seeds 1--100) as the main-text figures, except where noted.}

\subsection{Hyperparameters}
\rev{\Cref{tab:supp-hyperparams} lists the loss weights and learning rates. We chose them by hand for each method and model, aiming for stable training, and did not tune them systematically. \Cref{sec:supp-sensitivity} shows how the results change when the most influential weights are varied.}

\begin{table}[tbhp]
  \centering
  \footnotesize
  \caption{\rev{Hyperparameters used in the main text. A dash marks a weight that the method does not use.}}
  \label{tab:supp-hyperparams}
  \input{supp_figure/tab_supp_hyperparams.tex}
\end{table}

\subsection{Accuracy}
\rev{We measure accuracy by the relative $L^2$ error $\|\hat D-D_{\mathrm{true}}\|/\|D_{\mathrm{true}}\|$, both for each single fit and for the pointwise median curve that the main-text figures show (\cref{tab:supp-errors}). For reference, a constant $D$ equal to the true mean has error $0.33$. DTO has the smallest median-curve error in every condition. Its single fits are rough, and the wide bands in the main-text figures show how much they vary across trials. For every method, single fits are noisy: their typical error is larger than that of the constant reference, and the spatial profile is clearest in the median over trials. The constant reference has a small relative error because the mean accounts for most of the squared $L^2$ norm of $D_{\mathrm{true}}$. \Cref{tab:supp-errors,fig:supp-shape-scatter} therefore also give the Pearson correlation over $[0,1]$ between the recovered and true $D(x)$, which does not depend on the mean or scale of $\hat D$ and is undefined for a constant. Single fits correlate weakly with the true profile, in part because they are poor near the boundaries, where particles are sparse (\cref{sec:supp-failure}). The median curves correlate more strongly, except for BiLO on the Fickian $\sin 8\pi x$ profile, whose median curve does not follow the oscillation. Re-solving the model with each method's median $D$ and median $b_0$ gives densities within 5--18\% of the true density, although the median $D$ curves have errors of 21--74\%. The source strength $b_0$ is recovered to within about 10\% at $b_0=250$ and 5\% at $b_0=1000$ (median relative error).}

\begin{table}[tbp]
  \centering
  \scriptsize
  \setlength{\tabcolsep}{2pt}
  \caption{\rev{Relative $L^2$ error of $D$ and Pearson correlation of $D$ with $D_{\mathrm{true}}$ over $[0,1]$, each for single fits (median [10th, 90th percentile] over trials) and for the median curve, and relative error of $b_0$, for the five test conditions of the main text. The degenerate fits omitted from the main-text figures are excluded here too. Shading marks the best value among the three methods in each column (smallest error, largest correlation).}}
  \label{tab:supp-errors}
  \resizebox{\linewidth}{!}{\input{supp_figure/tab_supp_errors.tex}}
\end{table}

\begin{figure}[tbp]
  \centering
  \includegraphics[width=0.9\linewidth]{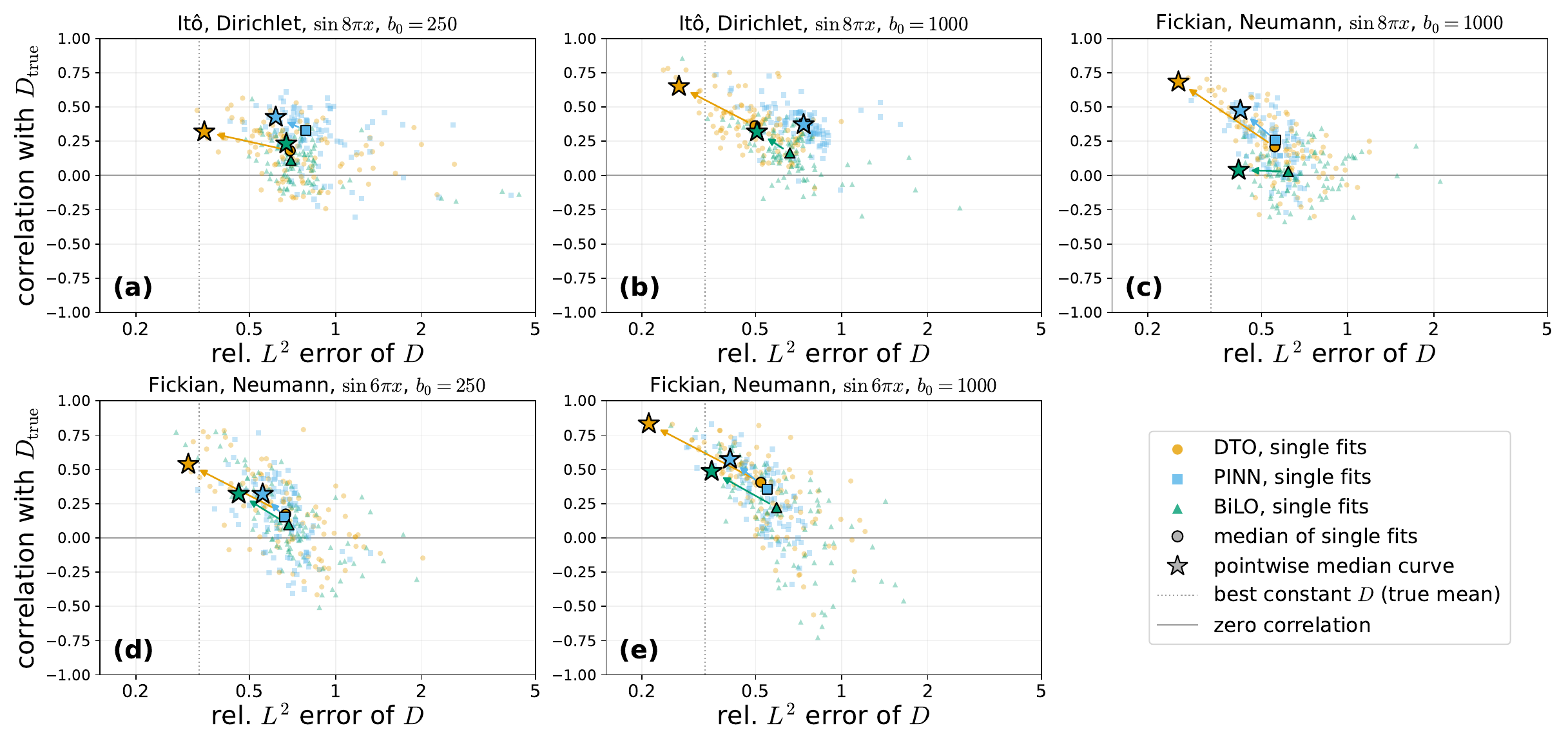}
  \caption{\rev{Correlation with $D_{\mathrm{true}}$ against relative $L^2$ error of $D$ for single fits in the five conditions of \cref{tab:supp-errors}, with the degenerate fits excluded. Larger markers give the medians over single fits, and arrows lead to the error and correlation of the pointwise median curve (star) shown in the main-text figures. The $L^2$ axis is logarithmic, and the dotted line is the error of the best constant $D$, equal to the true mean.}}
  \label{fig:supp-shape-scatter}
\end{figure}

\begin{figure}[tbp]
  \centering
  \includegraphics[width=0.9\linewidth]{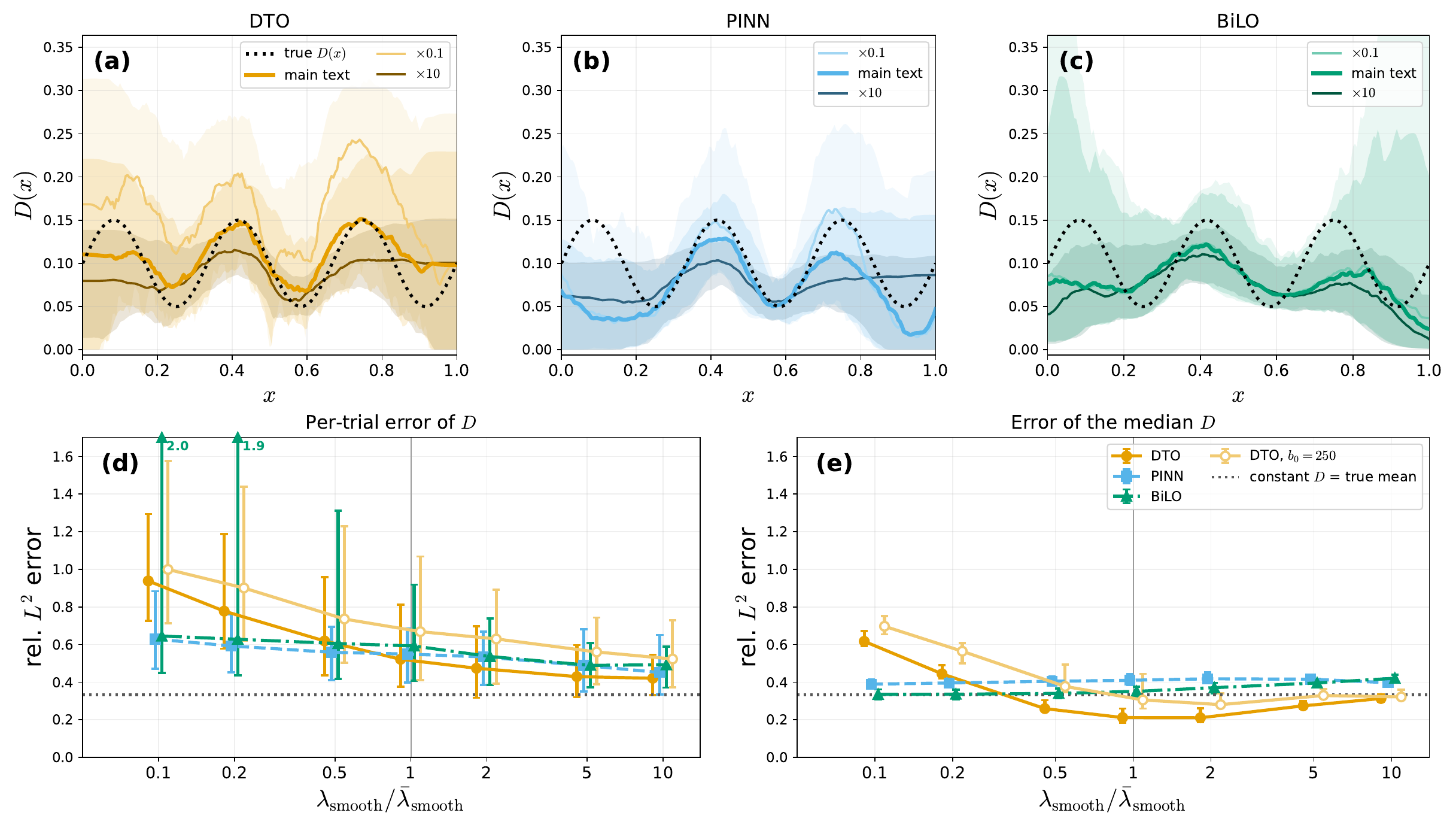}
  \caption{\rev{Sensitivity to the smoothness weight for Fickian dynamics with Neumann boundaries, $D_{\mathrm{true}}(x)=0.10+0.05\sin(6\pi x)$, and $b_0=1000$. (\textbf{a}--\textbf{c})~Median recovered $D(x)$ with 10th--90th percentile bands at $0.1\bar\lambda_{\mathrm{smooth}}$, $\bar\lambda_{\mathrm{smooth}}$ (main text), and $10\bar\lambda_{\mathrm{smooth}}$. (\textbf{d})~Per-trial relative $L^2$ error of $D$, median and 10th--90th percentile. (\textbf{e})~Error of the median curve, with bootstrap 10th--90th percentile bars. Open symbols in (d,~e) are DTO at $b_0=250$. The dotted line is a constant $D$ equal to the true mean.}}
  \label{fig:supp-smoothness}
\end{figure}

\subsection{Sensitivity to regularization}
\label{sec:supp-sensitivity}
\rev{We reran the condition of \cref{fig:Combined_Fickian_lowfreq}(c,d) with the smoothness weight of every method scaled by factors from $0.1$ to $10$ relative to its main-text value $\bar\lambda_{\mathrm{smooth}}$ (\cref{fig:supp-smoothness}). Stronger smoothing lowers the error of single fits by making them flatter, but it also flattens the median curve. Weaker smoothing makes single fits oscillate. The median-curve error is smallest near the main-text value for DTO and changes little for PINN and BiLO. At $b_0=250$, DTO behaves the same way, but its best median curve is only slightly better than the constant reference. The smoothness weight has little effect on $b_0$.}

\rev{PINN also depends on the weight $w_{\mathrm{res}}$ of the PDE residual (\cref{fig:supp-pinn-wres}). With small $w_{\mathrm{res}}$, $D$ and $\hat u$ can collapse toward zero over part of the domain, lowering the loss and increasing $b_0=N/\int\hat u$. Collapse becomes less frequent as $w_{\mathrm{res}}$ grows, but beyond about $20\bar w_{\mathrm{res}}$ the residual term outweighs the others and PINN overestimates $b_0$. The main-text value is at the low end of the range where $b_0$ is reliable.}

\begin{figure}[tbp]
  \centering
  \includegraphics[width=\linewidth]{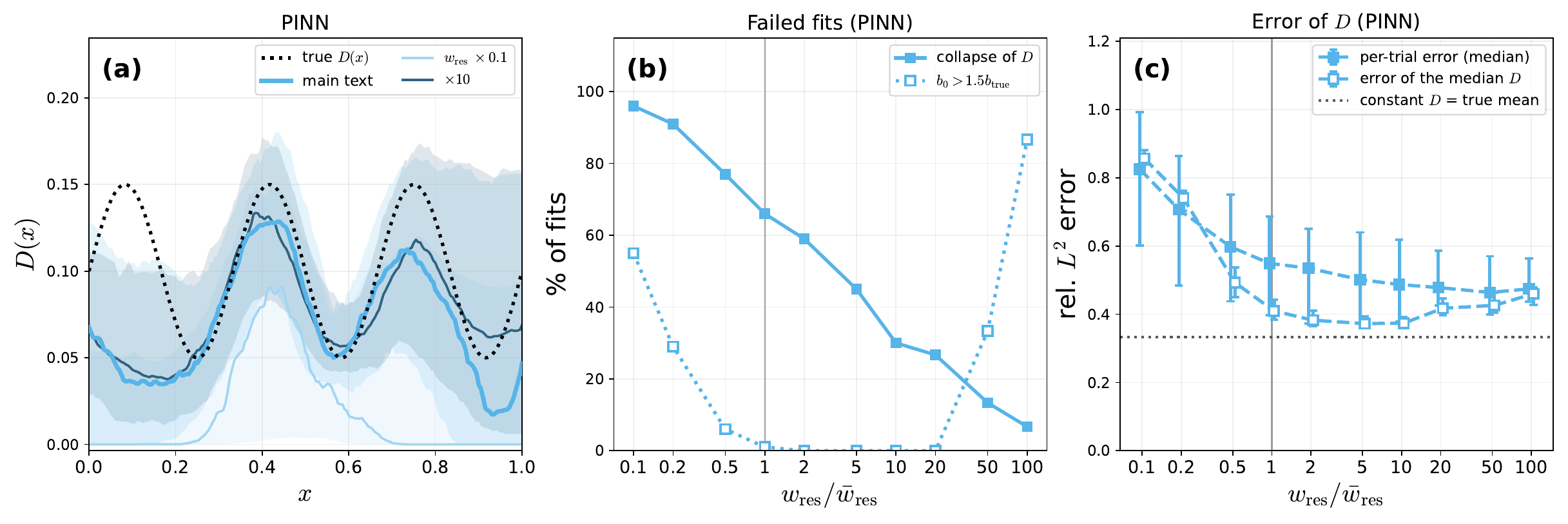}
  \caption{\rev{PINN sensitivity to the residual weight $w_{\mathrm{res}}$, for the condition of \cref{fig:supp-smoothness}. (\textbf{a})~Median $D(x)$ and bands at $0.1\bar w_{\mathrm{res}}$, $\bar w_{\mathrm{res}}$, and $10\bar w_{\mathrm{res}}$. (\textbf{b})~Percentage of fits in which $D$ collapses ($D<\langle D_{\mathrm{true}}\rangle/5$ on at least 10\% of the domain) or $b_0>1.5\,b_{\mathrm{true}}$. (\textbf{c})~Per-trial error (median and 10th--90th percentile) and error of the median curve. Points at $20$--$100\,\bar w_{\mathrm{res}}$ use 15 trials (seeds 1--15); the others use 100.}}
  \label{fig:supp-pinn-wres}
\end{figure}

\begin{figure}[tbp]
  \centering
  \includegraphics[width=0.72\linewidth]{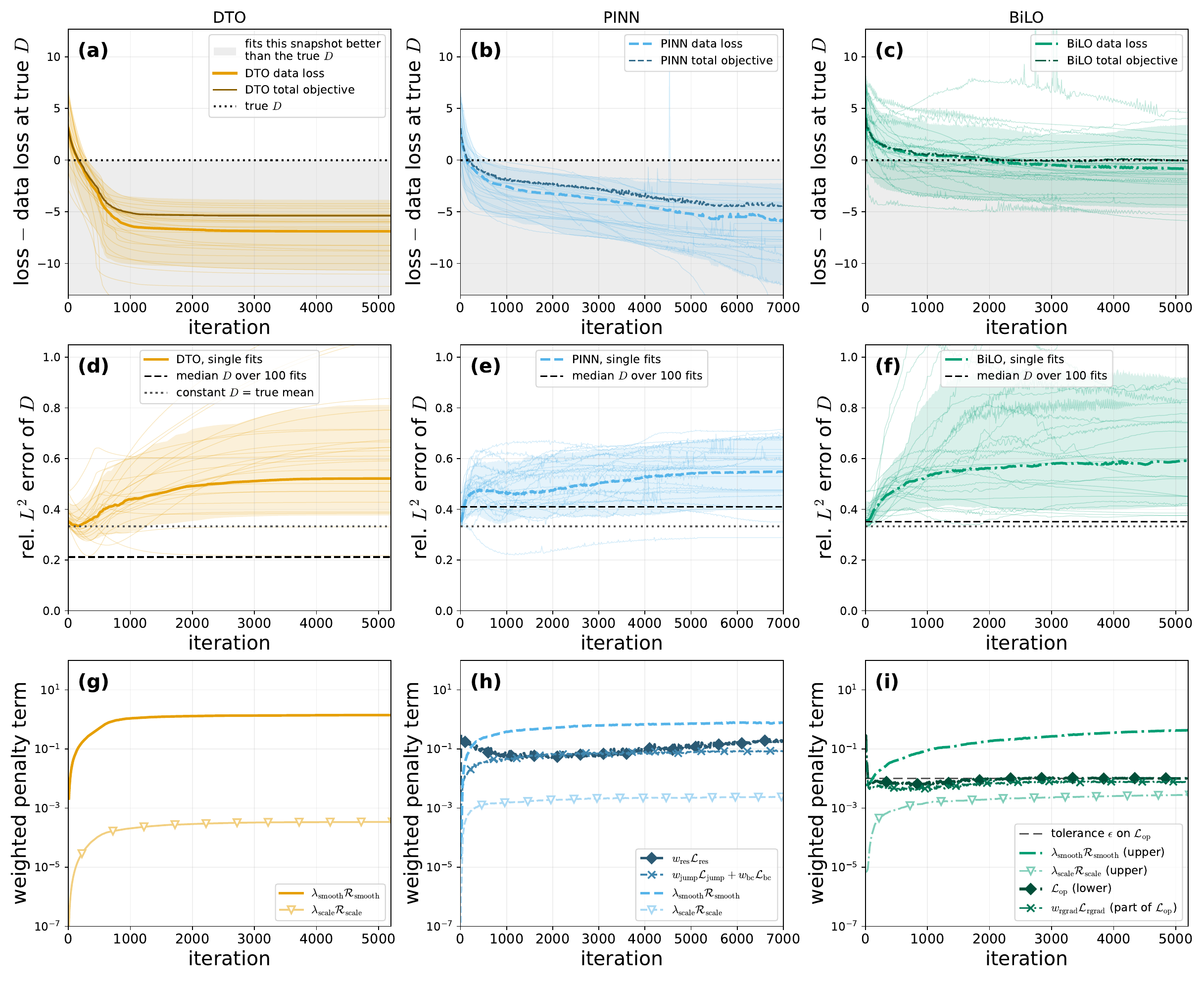}
  \caption{\rev{Training histories for the condition of \cref{fig:supp-smoothness} at the main-text settings. Thin lines are single trials, thick lines medians, and bands 10th--90th percentiles. (\textbf{a}--\textbf{c})~Data loss and total objective minus the data loss of the true $D$ on the same snapshot; shading marks where the fitted $D$ gives a lower data loss than the true $D$ on the same snapshot. (\textbf{d}--\textbf{f})~Relative $L^2$ error of $D$; dashed lines give the error of the median over the 100 final fits, and the dotted line the constant-$D$ reference. (\textbf{g}--\textbf{i})~Median weighted penalty terms; for BiLO the dashed line is the lower-level tolerance $\epsilon$.}}
  \label{fig:supp-training}
\end{figure}

\subsection{Training}
\rev{\Cref{fig:supp-training} shows the training histories for the same condition at the main-text settings. Because each snapshot is a random draw, the true $D$ need not minimize the data loss for that snapshot. For every method, the fitted $D$ has a lower final data loss than the true $D$ in at least half the trials. The error of $D$ is often lowest early in training and then rises. The weighted penalty terms level off well before training stops.}

\rev{To check how well the neural methods solve the PDE, we re-solved the steady state with the finite-volume solver for each fitted $D$ and compared the result with the network solution. The median relative $L^2$ difference is $0.059$ for PINN and $0.008$ for BiLO.}

\rev{This PDE error also helps explain why PINN slightly overestimates $b_0$ in \cref{fig:Combined_Fickian_lowfreq}. With reflecting boundaries, integrating the steady state for an exact unit-source response gives $\gamma\int\hat u=1$, so the projected estimate is $b^\ast=N/\int\hat u=\gamma N$, which DTO returns exactly. PINN satisfies the PDE and boundary conditions only approximately, and since $b_0/(\gamma N)=1/(\gamma\int\hat u)$, any error in its mass shows up directly in $b_0$.}

\begin{figure}[tbp]
  \centering
  \includegraphics[width=\linewidth]{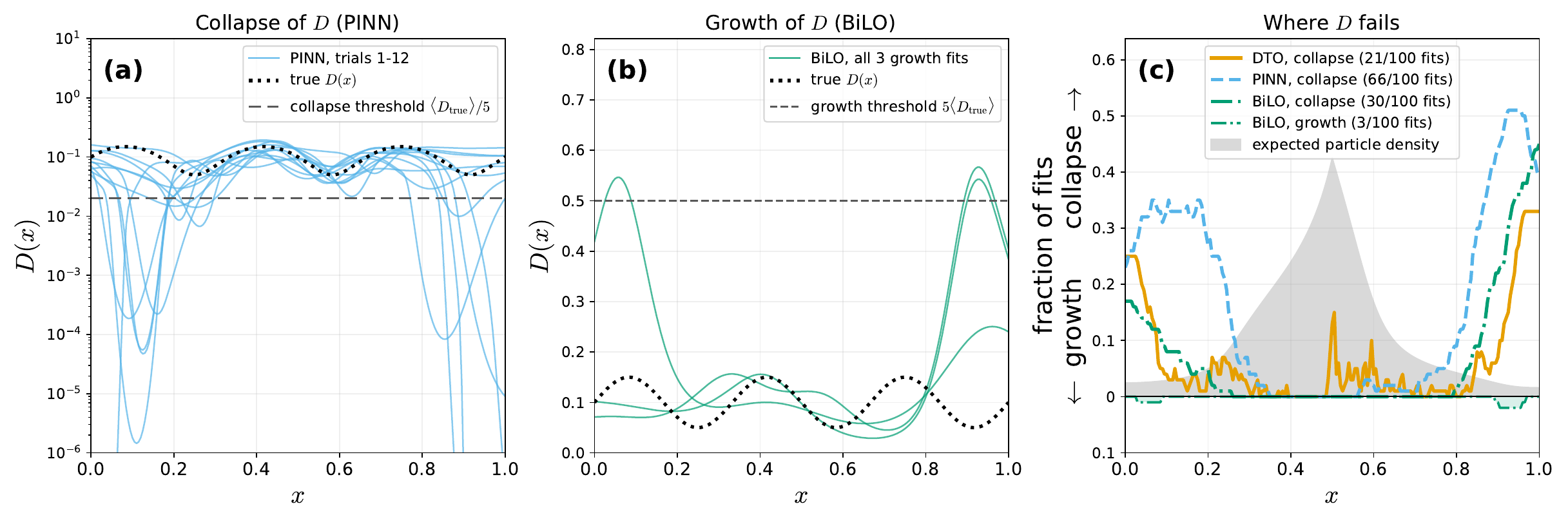}
  \caption{\rev{Failure modes at the main-text settings, for the condition of \cref{fig:supp-smoothness}. (\textbf{a})~PINN $D(x)$ for trials 1--12 on a log scale, with the collapse threshold. (\textbf{b})~The 3 BiLO fits with $D>5\langle D_{\mathrm{true}}\rangle$. (\textbf{c})~Fraction of fits that collapse (upward) or grow (downward) at each $x$, with the expected particle density in grey.}}
  \label{fig:supp-failure-modes}
\end{figure}

\subsection{Failure modes}
\label{sec:supp-failure}
\rev{Even at the main-text settings, some single fits fail in recognizable ways (\cref{fig:supp-failure-modes}). PINN estimates of $D$ can collapse toward zero over part of the domain, and BiLO estimates can grow large near a reflecting boundary. Both happen mostly near the boundaries, where particles are sparse. These failures have little effect on the median curves. A few fits give degenerate solutions with an unusually large $b_0$ or a nearly flat, unusually large $D$. We omit these fits from the recovery figures, and keeping them barely changes the median curves.}

\FloatBarrier

\printbibliography

\end{document}

%% file: supp_figure/tab_supp_hyperparams.tex
\begin{tabular}{l ccc ccc}
\toprule
 & \multicolumn{3}{c}{It\^o, Dirichlet} & \multicolumn{3}{c}{Fickian, Neumann} \\
\cmidrule(lr){2-4}\cmidrule(lr){5-7}
 & DTO & PINN & BiLO & DTO & PINN & BiLO \\
\midrule
$\lambda_{\mathrm{smooth}}$ & 10 & 0.05 & 2 & 1 & 2 & 1 \\
$\lambda_{\mathrm{scale}}$ & 1 & 1 & 1 & 0.1 & 1 & 1 \\
$w_{\mathrm{res}}$ & -- & 150 & -- & -- & 350 & -- \\
$w_{\mathrm{jump}}$ & -- & 20 & 1 & -- & 10 & 10 \\
$w_{\mathrm{bc}}$ & -- & -- & -- & -- & 1 & 1 \\
$w_{\mathrm{rgrad}}$ & -- & -- & $10^{-3}$ & -- & -- & 0.05 \\
lower-level tolerance $\epsilon$ & -- & -- & 0.01 & -- & -- & 0.01 \\
learning rate, $D$ (pretrain) & -- & -- & $10^{-4}$ & -- & -- & $10^{-3}$ \\
learning rate, $\hat u$ or $v$ (pretrain) & -- & $10^{-4}$ & $10^{-4}$ & -- & $10^{-3}$ & $10^{-3}$ \\
learning rate, $D$ & 0.01 & $10^{-3}$ & $5\times10^{-3}$ & $10^{-4}$ & 0.01 & $10^{-4}$ \\
learning rate, $\hat u$ or $v$ & -- & $5\times10^{-4}$ & 0.01 & -- & $10^{-4}$ & $10^{-3}$ \\
\bottomrule
\end{tabular}

%% file: supp_figure/tab_supp_errors.tex
\begin{tabular}{ll r l ccccc}
\toprule
 & & & & \multicolumn{2}{c}{error of $D$} & \multicolumn{2}{c}{correlation of $D$} & error of $b_0$ \\
\cmidrule(lr){5-6}\cmidrule(lr){7-8}\cmidrule(lr){9-9}
Model & $D_{\mathrm{true}}-0.1$ & $b_0$ & Method & per trial & median $D$ & per trial & median $D$ & per trial \\
\midrule
It\^o, & $0.05\sin 8\pi x$ & 250 & DTO & \cellcolor{mint}0.69 [0.47, 1.28] & \cellcolor{mint}0.35 & 0.18 [$-0.11$, 0.40] & 0.32 & 0.109 [0.022, 0.280] \\
Dirichlet & &  & PINN & 0.79 [0.61, 1.24] & 0.62 & \cellcolor{mint}0.33 [$-0.08$, 0.49] & \cellcolor{mint}0.43 & 0.095 [0.018, 0.214] \\
 & &  & BiLO & 0.70 [0.58, 0.85] & 0.67 & 0.11 [$-0.09$, 0.30] & 0.23 & \cellcolor{mint}0.085 [0.020, 0.254] \\
 & & & constant $D$ & 0.33 & 0.33 & -- & -- & -- \\
\midrule
It\^o, & $0.05\sin 8\pi x$ & 1000 & DTO & \cellcolor{mint}0.50 [0.36, 0.69] & \cellcolor{mint}0.27 & 0.36 [0.15, 0.65] & \cellcolor{mint}0.65 & \cellcolor{mint}0.052 [0.008, 0.125] \\
Dirichlet & &  & PINN & 0.74 [0.54, 0.87] & 0.74 & \cellcolor{mint}0.37 [0.22, 0.52] & 0.37 & 0.056 [0.010, 0.131] \\
 & &  & BiLO & 0.66 [0.48, 0.80] & 0.51 & 0.16 [$-0.06$, 0.38] & 0.32 & 0.053 [0.008, 0.130] \\
 & & & constant $D$ & 0.33 & 0.33 & -- & -- & -- \\
\midrule
Fickian, & $0.05\sin 8\pi x$ & 1000 & DTO & \cellcolor{mint}0.56 [0.38, 0.82] & \cellcolor{mint}0.26 & 0.21 [$-0.07$, 0.57] & \cellcolor{mint}0.68 & \cellcolor{mint}0.045 [0.010, 0.120] \\
Neumann & &  & PINN & \cellcolor{mint}0.56 [0.42, 0.66] & 0.42 & \cellcolor{mint}0.26 [$-0.02$, 0.50] & 0.48 & 0.055 [0.011, 0.167] \\
 & &  & BiLO & 0.62 [0.49, 0.96] & 0.42 & 0.03 [$-0.18$, 0.26] & 0.04 & 0.048 [0.009, 0.123] \\
 & & & constant $D$ & 0.33 & 0.33 & -- & -- & -- \\
\midrule
Fickian, & $0.05\sin 6\pi x$ & 250 & DTO & 0.67 [0.41, 1.07] & \cellcolor{mint}0.31 & \cellcolor{mint}0.17 [$-0.21$, 0.53] & \cellcolor{mint}0.54 & 0.100 [0.020, 0.240] \\
Neumann & &  & PINN & \cellcolor{mint}0.66 [0.48, 0.79] & 0.56 & 0.15 [$-0.19$, 0.44] & 0.32 & \cellcolor{mint}0.097 [0.020, 0.232] \\
 & &  & BiLO & 0.69 [0.48, 1.04] & 0.46 & 0.10 [$-0.19$, 0.54] & 0.32 & 0.098 [0.020, 0.239] \\
 & & & constant $D$ & 0.33 & 0.33 & -- & -- & -- \\
\midrule
Fickian, & $0.05\sin 6\pi x$ & 1000 & DTO & \cellcolor{mint}0.52 [0.38, 0.81] & \cellcolor{mint}0.21 & \cellcolor{mint}0.41 [$-0.06$, 0.70] & \cellcolor{mint}0.83 & \cellcolor{mint}0.045 [0.010, 0.120] \\
Neumann & &  & PINN & 0.55 [0.40, 0.68] & 0.41 & 0.35 [$-0.01$, 0.66] & 0.57 & 0.060 [0.010, 0.151] \\
 & &  & BiLO & 0.59 [0.41, 0.92] & 0.35 & 0.22 [$-0.33$, 0.60] & 0.49 & \cellcolor{mint}0.045 [0.011, 0.119] \\
 & & & constant $D$ & 0.33 & 0.33 & -- & -- & -- \\
\bottomrule
\end{tabular}